\documentclass[dvips,generic,10pt]{imsart}

\setattribute{journal}{name}{}
\renewcommand\baselinestretch{1.5}

\RequirePackage[OT1]{fontenc}
\RequirePackage{amsthm,amsmath}
\RequirePackage{natbib}
\RequirePackage[colorlinks,citecolor=blue,urlcolor=blue,breaklinks=true]{hyperref}
\RequirePackage{hypernat}
\RequirePackage{comment}
\RequirePackage{pzccal}
\RequirePackage{graphicx,subfigure,latexsym,amssymb}
\RequirePackage{float,epsfig,multirow,rotating}
\RequirePackage{upgreek,wrapfig}

\arxiv{math.PR/0000000}
\startlocaldefs
\newcommand {\ctn}{\citet} 
\newcommand {\ctp}{\citep}       

\newcommand{\bPhi}{\boldsymbol{\Phi}}

\newcommand{\bSigma}{\boldsymbol{\Sigma}}

\newcommand{\bPsi}{\boldsymbol{\Psi}}

\newcommand{\bOmega}{\boldsymbol{\Omega}}

\newcommand{\bD}{\boldsymbol{D}}

\newcommand{\bV}{\boldsymbol{V}}
\newcommand{\bG}{\boldsymbol{G}}

\newcommand{\bI}{\boldsymbol{I}}

\newcommand{\bE}{\boldsymbol{E}}

\newcommand{\bR}{\boldsymbol{R}}

\newcommand{\bx}{\boldsymbol{x}}
\newcommand{\bX}{\boldsymbol{X}}

\newcommand{\bY}{\boldsymbol{Y}}
\newcommand{\bZ}{\boldsymbol{Z}}
\newcommand{\bz}{\boldsymbol{z}}

\newcommand{\bL}{\boldsymbol{L}}

\numberwithin{equation}{section}
\theoremstyle{plain}
\newtheorem{theorem}{Theorem}[section]
\newtheorem{definition}{Definition}[section]
\newtheorem{remark}{Remark}[section]
\newtheorem{proposition}{Proposition}[section]

\endlocaldefs

\begin{document}
\renewcommand\baselinestretch{1.5}

\begin{frontmatter}
\title{Correlation between Multivariate Datasets, from Inter-Graph Distance
  computed using Graphical Models Learnt With Uncertainties}
\runtitle{With-uncertainty Learning of Graph $\&$ Correlation Structure}

\begin{aug}
\author{
{\fnms{Kangrui} \snm{Wang}\thanksref{t2,m2}\ead[label=e2]{rorschach.kangrui@gmail.com}}
\and
{\fnms{Dalia} \snm{Chakrabarty}\thanksref{t1,m1}\ead[label=e1]{d.chakrabarty@lboro.ac.uk}}
},
\thankstext{t2}{Postdoctoral Research Associate, Alan Turing Institute} 
\thankstext{t1}{Lecturer in Statistics, Department of Mathematical Sciences,
  Loughborough University}

\runauthor{Wang $\&$ Chakrabarty}

\affiliation{Alan Turing Institute, Loughborough University}

\address{\thanksmark{m2} Alan Turing Institute\\
British Library, 96 Euston Road, \\
London NW1 2DB,\\
U.K.\\
\printead*{e2}
}

\address{\thanksmark{m1} Department of Mathematical Sciences\\
Loughborough University\\
Loughborough LE11 3TU,
U.K.\\
\printead*{e1}
}

\end{aug}

\begin{abstract} {We present a method for simultaneous Bayesian learning of
    the correlation matrix and graphical model of a multivariate dataset,
    along with uncertainties in each, to subsequently compute distance between
    the learnt graphical models of a pair of datasets, using a new metric that
    approximates an uncertainty-normalised Hellinger distance between the
    posterior probabilities of the graphical models given the respective
    dataset; correlation between the pair of datasets is then computed as a
    corresponding affinity measure. We achieve a closed-form likelihood of the
    between-columns correlation matrix by marginalising over the between-row
    matrices. This between-columns correlation is updated first, given the
    data, and the graph is then updated, given the partial correlation matrix
    that is computed given the updated correlation, allowing for learning of
    the 95$\%$ Highest Probability Density credible regions of the correlation
    matrix and graphical model of the data. Difference made to the learnt
    graphical model, by acknowledgement of measurement noise, is demonstrated
    on a small simulated dataset, while the large human disease-symptom
    network--with $>8,000$ nodes--is learnt using real data. Data on
    vino-chemical attributes of Portuguese red and white wine samples are
    employed to learn with-uncertainty graphical model of each dataset, and
    subsequently, the distance between these learnt graphical models. }
\end{abstract}



\end{frontmatter}

\renewcommand\baselinestretch{1.0}
{
\section{Introduction}
\label{sec:intro}
\noindent
Graphical models of complex, multivariate datasets, manifest intuitive
illustrations of the correlation structures of the data, and are of interest
in different disciplines \ctp{whittaker, large_net, plos2007, carvalhowest,
  dipankar}. Much work has been undertaken to study the correlation structure
of a multivariate dataset comprising multiple measured values of a
vector-valued observable, by modelling the joint probability distribution of a
set of such observable values, as matrix-normal \ctp{multigraph, west_2016, wangwest}. In this paper, we simultaneously learn the partial correlation
structure and graphical model of a multivariate dataset, while making
inference on uncertainties of each, and acknowledge measurement errors in our
learning--with the ultimate aim of computing the distance between (posterior
probability distributions of) the learnt pair of graphical models of
respective datasets. Such distance informs us about the possible independence of the
datasets, generated under different environmental conditions. To this effect,
we undertake inference with Metropolis-within-Gibbs-based Bayesian inference
\ctp{Robert04}, on the correlation matrix given the data, and on the graph
given the updated correlation.

Objective and comprehensive uncertainties on the Bayesianly learnt graphical
model of given multivariate data, are sparsely available in the
literature. Such uncertainties can potentially be very useful in informing us
about the range of models that describe the partial correlation structure of
the data at hand. \ctn{madiganraftery} discuss a method for computing model
uncertainties by averaging over a set of identified models, and they advance
ways for the computation of the posterior model probabilities, by taking
advantage of the graphical structure, for two classes of considered models,
namely, the recursive causal models \ctp{carlincausal} and the decomposable
loglinear models \ctp{goodman}. This method allows them to select the ``best
models'', while accounting for model uncertainty. Our method on the other
hand, provides a direct and well-defined way of learning uncertainties of the
graphical model of a given multivariate data. At every update of our learning
of the graphical structure of the data, the graph is updated; graphs thus
learnt, if identified to lie within an identified range of values of the
posterior probability of the graph, comprise the uncertainty-included
graphical model of the data (Section~\ref{sec:95}). In addition, our method permits incorporation of
measurement errors into the learning of the graphical model, and permits fast
learning of large networks (Section~\ref{sec:disease}).

However, we wish to extend such learning to higher-dimensional
data, for example, to a dataset that is cuboidally-shaped, given that it
comprises multiple measurements of a matrix-valued observable.  \ctn{hoff,
  xu2012}; Wang $\&$ Chakrabarty ({\url{https://arxiv.org/abs/1803.04582}}), advance methods to
learn the correlation in high-dimensional data in general. For a
rectangularly-shaped multivariate dataset, the pioneering work by
\ctn{wangwest} allows for the learning of both the between-rows and
between-columns covariance matrices, and therefore, of two graphical
models. \ctn{multigraph} extend this approach to high-dimensional data.
However, a high-dimensional graph showing the correlation structure amongst
the multiple components of a general hypercuboidally-shaped dataset, is not
easy to visualise or interpret. Instead, in this paper, we treat the
high-dimensional data as built of correlated rectangularly-shaped slices,
given each of which, the between-columns (partial) correlation structure and
graphical model are Bayesianly learnt, along with uncertainties, subsequent to
our closed-form marginalisation over all between-rows correlation matrices
(in Section~\ref{sec:corr}, unlike in the work of \ctn{wangwest}). By invoking the uncertainties learnt
in the graphical models, we advance a new inter-graph distance
metric (Section~\ref{sec:hell}), based on the Hellinger distance \ctp{matu,juq} between the posterior probability densities of the pair of
graphical models that are learnt given the respective pair of such
rectangularly-shaped data slices. We use a proposed affinity measure to
infer on the correlation between the datasets (Section~\ref{sec:itisdist}). For example, by computing the
pairwise inter-graph distance between posterior probability densities of each
learnt pair of graphs, we can avoid the inadequacy of trying to capture
spatial correlations amongst sets of multivariate observations, by ``computing
partial correlation coefficients and by specifying and fitting more complex
graphical models'', as was noted by \ctn{guiness}. In fact, our method offers
the inter-graph distance for two differently sized datasets. 

Importantly, we will demonstrate below that it is the learning of
uncertainties in graphical models, that allows for the pursuit of the
inter-graph distance.
 
Our learnt graphical model of the given data, comprises a set of
random inhomogeneous graphs \ctp{book_graph} that lie within the
credible regions that we define, where each such graph is a
generalisation of a Binomial graph.
We do not make inference on the graph
(writing its posterior) clique-by-clique, and neither are we reliant
on the closed-form nature of the posteriors to sample from. In other
words, we do not need to invoke conjugacy to affect our
learning--either of the partial correlation structure of the data or
of the graphical model. Often, in Bayesian learning of Gaussian
undirected graphs, a Hyper-Inverse-Wishart prior is typically
imposed on the covariance matrix of the data, as this then allows for
a Hyper-Inverse-Wishart posterior of the covariance, which in turn
implies that the marginal posterior of of any clique is
Inverse-Wishart--a known, closed-form density \ctp{dawidlauritzen93,
  lauritzen96}.
Inference is then rendered easier, than when posterior sampling from a
non-closed form posterior needs to be undertaken, using numerical
techniques such as MCMC. Now, if the graph is not decomposable, and a
Hyper-Inverse-Wishart prior is placed on the covariance matrix, the
resulting Hyper-Inverse-Wishart joint posterior density that can be
factorised into a set of Inverse-Wishart densities, cannot be
identified as the clique marginals. Expressed differently, the clique
marginals are not closed-form when the graph is not decomposable.
However, this is not a worry in our learning, i.e. we can undertake
our learning irrespective of the validity of decomposability.

This paper is organised as follows. The following section deliberates upon the
methodology that we advance, including the closed-form likelihood of the
between-column correlation matrix of the data at hand, and definition of the
uncertainties on the learnt graphical model. The method of computing the
inter-graph distance that invokes such learnt uncertainties, is then discussed
in Section~\ref{sec:hell}. Section~\ref{sec:real} presents the emiprical
illustration on 2 real datasets, with the distance between the learnt,
with-uncertainty graphical models of these 2 data, discussed in
Section~\ref{sec:real_hell}.  In Section~\ref{sec:disease}, we learn the
graphical model of a real, highly multivariate, dataset, namely the human
disease-phenotype dataset, and compare our results with those reported earlier
\ctp{hsg}. The paper is rounded up with a section that
summarises the main findings and the conclusions.
The attached Supplementary Materials elaborate on certain aspects of our
work. This includes comparison of results obtained by using our
method with existing and independently obtained results, relevant to a pair of
real datasets that we illustrate our methodology on in this paper (Sections~4
and 6 of the Supplementary Material), and importantly, detailed model checking
is discussed in Section~2 of the Supplementary Material.



\section{Learning correlation matrix and graphical model given data, using
  Metropolis-within-Gibbs}
\label{sec:corr}
\noindent
Let $\bX\in{\cal X}\subseteq{\mathbb R}^p$ be a $p$-dimensional observed vector, with $\bX=(X_1,\ldots,X_p)^T$. Let there be $n$ measurements of $X_j$, $j=1,\ldots,p$, so that the $n\times p$-dimensional matrix ${\bf D}=[x_{ij}]_{i=1;j=1}^{n;p}$ is the data that comprises $n$ measurements of the $p$-dimensional observable $\bX$. Let the $i$-th realisation of $\bX$ be $\bx_i$, $i=1,\ldots,n$. 
We standardise the variable $X_j$ ($j=1,\ldots,p$) by its empirical mean
and standard deviation, into $Z_j$, s.t. the standardised version ${\bf D}_S$
of data ${\bf D}$ comprises $n$ measurements of the $p$-dimensional vector
$\bZ=(Z_1,\ldots,Z_p)^T$. Thus, 
$z_{ij}=\displaystyle{\frac{x_{ij}-\bar{x}_j}{\Upsilon_j}}$, where
$\bar{x}_j:=\displaystyle{\frac{\sum\limits_{i=1}^n x_{ij}}{n}}$ and
$\Upsilon_j^2 := \displaystyle{\frac{\sum\limits_{i=1}^n x_{ij}^2}{n}
  -\left(\frac{\sum\limits_{i=1}^n x_{ij}}{n}\right)^2}$. 
The $n\times p$-dimensional matrix ${\bf D}_S=[z_{ij}]$.
Then we model the joint probability of a set of 
measurements of $\bZ$, (such as the set of $n$ that comprises the
standardised data ${\bf D}_S$), to be matrix-normal with zero-mean, i.e.
$$\{\bz_1,\ldots,\bz_n\}\sim{\cal MN}({\bf 0}, \bSigma_R^{(S)}, \bSigma_C^{(S)}),$$
i.e. the likelihood of the covariance matrices $\bSigma_R^{(S)}$ and $\bSigma_C^{(S)}$, given data ${\bf D}_S$, is matrix-normal: 
\begin{equation}
\ell(\bSigma_R^{(S)}, \bSigma_C^{(S)}\vert {\bf D}_S) =
\displaystyle{
\frac{1}{(2\pi)^{\frac{np}{2}}
|\bSigma_C^{(S)}|^{\frac{p}{2}}
|\bSigma_R^{(S)}|^{\frac{n}{2}}}\times
\exp\left[
-\frac{1}{2}tr
\left\{(\bSigma_R^{(S)})^{-1} {\bf D}_S (\bSigma_C^{(S)})^{-1} ({\bf D}_S)^T\right\}
\right]
},
\label{eqn:matrix_normal_density}
\end{equation}
Here $\bSigma_R^{(S)}$ generates the covariance between the standardised variables $\bZ_i$ and $\bZ_{i^{/}}$, $i,i^{/}=1,\ldots,n$, (while 
$\bSigma_R$ generates the covariance between $\bX_i$ and $\bX_{i^{/}}$). In other words, $\bSigma_R^{(S)}$ generates the correlation between rows of the standardised data set ${\bf D}_S$. Similarly, $\bSigma_C^{(S)}$ generates the correlation between columns of ${\bf D}_S$. 


\begin{theorem}
\label{theo:marg}
The joint posterior probability density of the correlation matrices $\bSigma_C^{(S)}, \bSigma_R^{(S)}$, given the standardised data ${\bf D}_S$ is
$$
\left[\bSigma_C^{(S)}, \bSigma_R^{(S)}\vert{\bf D}_S\right] \propto \ell(\bSigma_R^{(S)}, \bSigma_C^{(S)}\vert {\bf D}_S) \left[\bSigma_C^{(S)}, \bSigma_R^{(S)}\right],
$$
where $\ell(\bSigma_R^{(S)}, \bSigma_C^{(S)}\vert {\bf D}_S)$ is the likelihood of $\bSigma_R^{(S)}, \bSigma_C^{(S)}$ given data ${\bf D}_S$. 
This can be marginalised over the $n\times n$-dimensional between-rows' correlation $\bSigma_R^{(S)}$, to yield
$$
[\bSigma_C^{(S)}\vert{\bf D}_S]\propto
\displaystyle{\frac{1}{c\left(\bSigma_C^{(S)}\right){\Big{\vert}} \bSigma_C^{(S)}{\Big{\vert}}^{p/2} {\Big{\vert}}{\bf D}_S (\bSigma_C^{(S)})^{-1} ({\bf D}_S)^T {\Big{\vert}}^{\frac{n+1}{2}}}},
$$
where the prior on $\bSigma_C^{(S)}$ is uniform; prior on $\bSigma_R^{(S)}$ is
the non-informative $\pi_0(\bSigma_R^{(S)})={\Big{\vert}}
\bSigma_R^{(S)}{\Big{\vert}}^{\alpha}$, $\alpha=\displaystyle{-\frac{n}{2}
  -1}$, and $\bSigma_C^{(S)}$ is assumed invertible. Here, $c\left(\bSigma_C^{(S)}\right)$ is a function of $\bSigma_C^{(S)}$ that normalises the likelihood. 
\end{theorem}

\begin{proof}
The joint posterior probability density of $\bSigma_C^{(S)}, \bSigma_R^{(S)}$, given data ${\bf D}_S$:
\begin{eqnarray}
\left[\bSigma_C^{(S)}, \bSigma_R^{(S)}\vert{\bf D}_S\right] 
&\propto& 
\displaystyle{
\ell\left(\bSigma_R^{(S)}, \bSigma_C^{(S)}\vert {\bf D}_S\right)
\left[\bSigma_C^{(S)}, \bSigma_R^{(S)}\right], \quad{\mbox{i.e.}}}\nonumber \\
\left[\bSigma_C^{(S)}, \bSigma_R^{(S)}\vert{\bf D}_S\right] &\propto& 
\displaystyle{
\frac{1}{(2\pi)^{\frac{np}{2}}
{\Big{|}}\bSigma_C^{(S)}{\Big{|}}^{\frac{p}{2}}
{\Big{|}}\bSigma_R^{(S)}{\Big{|}}^{\frac{n}{2}}}\times}\nonumber \\
&&\displaystyle{
\exp\left[
-\frac{1}{2}tr
\left\{(\bSigma_R^{(S)})^{-1} ({\bf D}_S) (\bSigma_C^{(S)})^{-1} 
({\bf D}_S)^T\right\}
\right]
{\Big{|}\bSigma_R^{(S)}\Big{|}^{-\frac{n}{2} -1}}},
\nonumber \\
&&
\end{eqnarray}
using the likelihood from Equation~\ref{eqn:matrix_normal_density}; using prior on 
$\bSigma_R^{(S)}$ to be
$\pi_0(\bSigma_R^{(S)})={\Big{\vert}}
\bSigma_R^{(S)}{\Big{\vert}}^{\alpha}$ where $\alpha=\displaystyle{-\frac{n}{2}
  -1}$; using prior on
$\bSigma_C^{(S)}$ to be uniform.

Marginalising $\bSigma_R^{(S)}$ out from the joint posterior
$\left[\bSigma_C^{(S)}, \bSigma_R^{(S)}\vert{\bf D}_S\right]$, we get:\\
$$
\left[\bSigma_C^{(S)}\vert{\bf D}_S\right] \propto
$$
\begin{equation}
\displaystyle{
\frac{1}{{\Big{|}}\bSigma_C^{(S)}{\Big{|}}^{\frac{p}{2}}}\times} 
\displaystyle{\int\limits_{{\cal R}}
\frac{1}{{\Big{|}}\bSigma_R^{(S)}{\Big{|}}^{\frac{n}{2}}}
{\Big{|}}\bSigma_R^{(S)}{\Big{|}}^{-\frac{n}{2}-1}\times
\exp\left[
-\frac{1}{2}tr
\left\{(\bSigma_R^{(S)})^{-1} {\bf D}_S (\bSigma_C^{(S)})^{-1} ({\bf D}_S)^T\right\}
\right]
d(\bSigma_R^{(S)})
}
\label{eqn:2}
\end{equation}
Here $\bSigma_R^{(S)}\in{\cal R}\subseteq{\mathbb R}^{(n\times n)}$. Now, 
\begin{enumerate}
\item[--]let $\bY:= (\bSigma_R^{(S)})^{-1}$. Then $d(\bSigma_R^{(S)}) = \vert \bY\vert^{-(n+1)} d\bY$ \ctp{mathai},
\item[--]let $\bV^{-1}:={\bf D}_S (\bSigma_C^{(S)})^{-1} ({\bf D}_S)^T$,
  $\Longrightarrow tr\left[(\bSigma_R^{(S)})^{-1} {\bf D}_S (\bSigma_C^{(S)})^{-1} ({\bf D}_S)^T\right] \equiv 
tr\left[\bV^{-1}\bY\right]$ (using commutativeness of trace), 
\end{enumerate}
so that in Equation~\ref{eqn:2}, we get
\begin{eqnarray}
\left[\bSigma_C^{(S)}\vert{\bf D}_S\right] &\propto&
\displaystyle{
\frac{1}{{\Big{|}}\bSigma_C^{(S)}{\Big{|}}^{\frac{p}{2}}}
\int\limits_{{\cal R}}
{|\bY|^{\frac{n}{2}}}
|\bY|^{\frac{n}{2}+1}\times
\exp\left[
-\frac{1}{2}tr
\left\{\bV^{-1}\bY\right\}
\right]
\vert \bY\vert^{-(n+1)} d\bY
}. \nonumber \\ 
&& 
\label{eqn:3}
\end{eqnarray}
The integral in the RHS of Equation~\ref{eqn:3} represents the unnormalised
Wishart $pdf$ $W_n(\bV, q)$, over all values of the random matrix
$\bY$, where the scale matrix and degrees of freedom of this
$pdf$ are $\bV$ and $q=n+1$ respectively, i.e. $q > n-1$. \\
Thus, integral in the RHS of Equation~\ref{eqn:3} is the integral of the unnormalised $pdf$ of $\bY\sim W_n(\bV, q)$, over the full support of $\bY\left(\equiv\left(\bSigma_R^{(S)}\right)^{-1}\right)$, \\
i.e. the integral in the RHS of Equation~\ref{eqn:3} is
the normalisation of this $pdf$:
$$2^{\frac{qn}{2}} \Gamma_n\left(\frac{q}{2}\right)
\vert\bV\vert^{\frac{q}{2}}\equiv $$
$$2^{\frac{(n+1)(n)}{2}} \Gamma_n\left(\frac{n+1}{2}\right) {\Big{|}}\left({\bf D}_S (\bSigma_C^{(S)})^{-1} ({\bf D}_S)^T\right)^{-1}{\Big{|}}^{\frac{n+1}{2}},$$
i.e. integral on RHS of Equation~\ref{eqn:3} is proportional to
${\Big{|}}\left({\bf D}_S (\bSigma_C^{(S)})^{-1} ({\bf
    D}_S)^T\right)^{-1}{\Big{|}}^{\frac{n+1}{2}},\quad{\mbox{i.e.}}$
\begin{equation}
\left[\bSigma_C^{(S)}\vert{\bf D}_S\right] \propto
\displaystyle{
\frac{1}{{\Big{|}}\bSigma_C^{(S)}{\Big{|}}^{\frac{p}{2}}}
{\Big{|}}\left({\bf D}_S (\bSigma_C^{(S)})^{-1} ({\bf
    D}_S)^T\right)^{-1}{\Big{|}}^{\frac{n+1}{2}}}
\label{eqn:notun}
\end{equation}

Now, if
${\bf D}_S (\bSigma_C^{(S)})^{-1} ({\bf D}_S)^T$ is invertible, ${\Big{|}}\left( {\bf D}_S (\bSigma_C^{(S)})^{-1} ({\bf D}_S)^T
\right)^{-1}{\Big{|}}^{\cdot}= {\Big{|}}{\bf D}_S (\bSigma_C^{(S)})^{-1} ({\bf
  D}_S)^T{\Big{|}}^{-{\cdot}}$.
\begin{enumerate}
\item[--]It is given that $\bSigma_C^{(S)}$ is invertible, i.e. $\left(\bSigma_C^{(S)}\right)^{-1}$ exists.
\item[--]The original dataset is examined to discard rows that are linear
transformations of each other, leading to data matrix ${\bf D}_S$, no two rows of which are linear transformations of each other
\end{enumerate}
$\Longrightarrow$ ${\bf D}_S (\bSigma_C^{(S)})^{-1} ({\bf D}_S)^T$ is
  positive definite, i.e. ${\bf D}_S (\bSigma_C^{(S)})^{-1} ({\bf D}_S)^T$ is
invertible,\\ $\Longrightarrow$
${\Big{|}}\left( {\bf D}_S (\bSigma_C^{(S)})^{-1} ({\bf D}_S)^T \right)^{-1}{\Big{|}}^{(n+1)/2}= {\Big{|}}{\bf D}_S (\bSigma_C^{(S)})^{-1} ({\bf D}_S)^T{\Big{|}}^{-{(n+1)/2}}$. 

Using this in Equation~\ref{eqn:notun}:
\begin{equation}
\left[\bSigma_C^{(S)}\vert{\bf D}_S\right] \propto {\Big{|}}{\bSigma}_C^{(S)}{\Big{|}}^{-{p/2}} {\Big{|}}{\bf D}_S (\bSigma_C^{(S)})^{-1} ({\bf
  D}_S)^T{\Big{|}}^{-{(n+1)/2}}.
\label{eqn:notun2}
\end{equation}
 
This posterior of the
between-columns correlation matrix $\bSigma_C^{(S)}$ given data ${\bf D}_S$,
is normalised over all possible datasets, where the possible datasets abide
by a column-correlation matrix of $\bSigma_C^{(S)}$, as:
\begin{equation}
c\left(\bSigma_C^{(S)}\right)=\displaystyle{\int\limits_{\cal Z} \ldots \int\limits_{\cal Z}
  \frac{1
}{{\Big{|}}\left( {\bf D}^{/} (\bSigma_C^{(S)})^{-1} ({\bf
      D}^{/})^T\right){\Big{|}}^{\frac{n^{/}+1}{2}}} dz_{11}^{/}dz_{11}^{/}\ldots dz_{n^{/}p}^{/}},
\label{eqn:norm}
\end{equation}
where ${\bf D}^{/}=[z_{ij}^{/}]_{i=1;j=1}^{i=n^{/}; j=p}$ is a dataset with $n^{/}$ rows and $p$ columns, comprising values of random standardised variables $Z_{ij}^{/}\in{\cal Z}$, simulated to bear between-column correlation matrix of $\bSigma_C^{(S)}$, s.t. ${\bf D}^{/} (\bSigma_C^{(S)})^{-1} ({\bf D}^{/})^T$ is positive definite $\forall {\bf D}^{/} \in {\cal D}$. 
Choosing the same number of rows for all choices of the
random data matrix ${\bf D}^{/}$, i.e. for a constant $n^{/}$, ${\cal
  D}\subseteq{\mathbb R}^{(n^{/}\times p)}$. Then 
$c\left(\bSigma_C^{(S)}\right) >0$ for any $\bSigma_C^{(S)}$.

Using this normalisation on the posterior of $\bSigma_C^{(S)}$ given ${\bf
  D}_S$, in Equation~\ref{eqn:notun2} we get 
\begin{equation}
\pi\left(\bSigma_C^{(S)}\vert{\bf D}_S\right) =
\displaystyle{
\frac{1}{c\left(\bSigma_C^{(S)}\right) {\Big{|}}\bSigma_C^{(S)}{\Big{|}}^{\frac{p}{2}}}
\frac{1}{{\Big{|}}\left( {\bf D}_S (\bSigma_C^{(S)})^{-1} ({\bf D}_S)^T\right){\Big{|}}^{\frac{n+1}{2}}}
},
\label{eqn:4}
\end{equation}
where $c\left(\bSigma_C^{(S)}\right)>0$ is defined in Equation~\ref{eqn:norm}.
\end{proof}

\begin{proposition}
{
An estimator of the normalisation ${\hat c}\left(\bSigma_C^{(S)}\right)$ of the posterior $\left[\bSigma_C^{(S)}\vert{\bf D}_S\right]$,
given 
in Equation~\ref{eqn:norm} is
$${\hat c}\left(\bSigma_C^{(S)}\right) = 
\displaystyle{{\mathbb E}_{Z^{/}_{n^{/}p}}\left[\ldots\left[{\mathbb E}_{Z^{/}_{11}}
\left[ \displaystyle{\frac{1} {{\Big{|}}\left( {\bf D}^{/} (\bSigma_C^{(S)})^{-1} ({\bf D}^{/})^T\right){\Big{|}}^{\frac{n^{/}+1}{2}}}}  \right]\right]\ldots\right]
}.$$
We substitute this difficult, sequential computing of expectations
w.r.t. distribution of each element of ${\bf D}^{/}$, by computation 
of the expectation w.r.t. the block ${\bf D}^{/}$ of these elements, where ${\bf D}^{/}$ abides by a column-correlation of $\bSigma_C^{(S)}$, i.e., we compute
$${\hat c}^{/}\left(\bSigma_C^{(S)}\right) = \displaystyle{{\mathbb E}_{{\bf D}^{/}_{S}}
\left[\displaystyle{\frac{1} {{\Big{|}}\left( {\bf D}^{/}
        (\bSigma_C^{(S)})^{-1} ({\bf
          D}^{/})^T\right){\Big{|}}^{\frac{n^{/}+1}{2}}}} \right]}.$$ 
We consider a between-columns correlation matrix $\bSigma_t$, and the sample of $k$ number of $n^{/}\times
p$-dimensional data sets $\{{\bf D}_1^{t/}, \ldots,{\bf D}_K^{t/}\}$,
s.t. ${\bf D}_k^{t/} (\bSigma_t)^{-1} ({\bf D}_k^{t/})^T$ is positive definite
$\forall k=1,\ldots,K$, at each $t$, the estimator of ${\hat
  c}^{/}\left(\bSigma_t\right)$ is 
\begin{equation}
{\hat{c}}_t := \displaystyle{
\frac{1}{K}
{\sum\limits_{k=1}^K 
\frac{1}{{\Big{|}}\left( {\bf D}_k^{t/} (\bSigma_t)^{-1} ({\bf
      D}_k^{t/})^T\right){\Big{|}}^{\frac{n^{/}+1}{2}}}}
}.
\label{eqn:est}
\end{equation}
}
\end{proposition}
Generation of a randomly sampled $n^{/}\times p$-sized data
set ${\bf D}_k^{t/}$, with column correlation $\bSigma_t$, is undertaken.


\subsection{Learning the graphical model}
\label{sec:graph}
\noindent
We perform Bayesian learning of the inhomogeneous, Generalised Binomial random
graph ${\mathbb G}(p, \bR)$, given the learnt $p\times p$-dimensional,
between-columns correlation matrix $\bSigma_C^{(S)}$, of the standardised data
set ${\bf D}_S:=(\bZ_1,\vdots,\ldots,\vdots,\bZ_p)^T$.  Here, the graph
${\mathbb G}(p, \bR)$, has the vertex set $\bV$ and the between-columns
partial correlation matrix $\bR$ of data ${\bf D}_S$, where 
$\bR=[R_{ij}]$, s.t. $R_{ij}$
takes the value $\rho_{ij}$, $i\neq j$, and $\rho_{ii}=1$. The vertex set is
$\bV=\{1,\ldots,p\}$ s.t. vertices $i,j\in \bV,\:i\neq j$, are joined by the
edge $G_{ij}$ that is a random binary variable taking values of $g_{ij}$,
where $g_{ij}$ is either 1 or 0, and is the $ij$-th element of the edge matrix
$\bG=[G_{ij}]$.

Given a learnt value of the between-columns correlation matrix $\bSigma_C^{(S)}$, to compute the value $\rho_{ij}$ of the partial correlation variable $R_{ij}$, we first invert $\bSigma_C^{(S)}$ to yield: $\bPsi:=\left(\bSigma_C^{(S)}\right)^{-1};\:\bPsi=[\psi_{ij}]$, s.t. 
\begin{equation}
R_{ij} = -\displaystyle{\frac{\psi_{ij}}{\sqrt{\psi_{ii} \psi_{jj}}}},\quad i\neq j,
\label{eqn:6}
\end{equation}
and $\rho_{ii}=1$ for $i=j$. 

The posterior probability density of the graph
${\mathbb G}(p, \bR)$ defined
for the edge matrix $\bG$, is given as
$$\pi(G_{11}, G_{12},\ldots G_{p\:p-1}\vert \bR) \propto \ell(G_{11}, G_{12},\ldots G_{p\:p-1}\vert \bR)\:\pi_0(G_{11}, G_{12},\ldots G_{p\:p-1}),$$
where $\pi_0(G_{11}, G_{12},\ldots G_{p\:p-1})$ is the prior probability
density on the edge parameters $\{G_{ij}\}_{i\neq j; i,j=1}^{p}$. We choose a prior on $G_{ij}$ that is $Bernoulli(0.5)$, i.e. $\pi_0(G_{11}, G_{12},\ldots G_{p\:p-1}) =\displaystyle{\prod\limits_{i,j=1; i\neq j}^p 0.5^{g_{ij}} 0.5^{1-g_{ij}} }$; thus, the prior is independent of the edge parameters. In applications marked by more information, we can resort to stronger priors.

$\ell(G_{12},\ldots, G_{1p},G_{23},\ldots,G_{2p},G_{34},\ldots,G_{p\:p-1}\vert
\bR)$ is the likelihood of the edge parameters, given the partial correlation
matrix $\bR$ (that is itself computed using the between-columns correlation
matrix $\bSigma_C^{(S)}$, learnt given ${\bf D}_S$, (see
Equation~\ref{eqn:4}). 
We choose to define this likelihood as a function of the (squared) Euclidean distance between
the ``observation'', i.e. the value of $R_{ij}$, and the unknown parameter
$G_{ij}$, with the squared distance normalised by a squared scale length, or
variance parameter $\sigma_{ij}^2$, for all relevant pairs  of nodes. Thus, the
unknown parameters in the model are the edge and variance parameters; in light
of these newly introduced variance parameters, we rewrite our likelihood as
$\ell(G_{12},\ldots,G_{1p},G_{23},\ldots,G_{2p},G_{34},\ldots,G_{p\:p-1},$\\ 
$\sigma_{12}^2,\ldots,\sigma_{1p}^2,\sigma_{23}^2,\ldots,\sigma_{2p}^2,\sigma_{34}^2,\ldots,\sigma_{p\:p-1}^2\vert\bR)$. Then
the constraints on the likelihood function suggest that likelihood increases (decreases) as distance between $R_{ij}$ and $G_{ij}$ decreases
(increases), and likelihood invariant to change of sign of
$R_{ij}-G_{ij}$. Given these constraints, we model our
likelihood of the edge and variance parameters, given $\bR$ as  
$$\ell\left(G_{12},\ldots,G_{1p},G_{23},\ldots,G_{2p},\ldots,G_{p\:p-1},\sigma_{12}^2,\ldots,\sigma_{1p}^2,\sigma_{23}^2,\ldots,\sigma_{2p}^2,\ldots,\sigma_{p\:p-1}^2\vert\bR\right) =$$
\begin{equation}
\displaystyle{
\prod\limits_{i\neq j; i,j=1}^{p} \frac{1}{\sqrt{2\pi\sigma_{ij}^2}}
\exp\left[-\frac{\left(G_{ij} - R_{ij}\right)^2}{2\sigma_{ij}^2} 
\right]}, 
\label{eqn:folded}
\end{equation}
where the variance parameters $\{\sigma_{ij}^2\}_{i\neq j;i,j=1}^p$ are indeed hyperparameters that are also learnt from the data; these variance parameters have uniform prior probabilities imposed on them.

\subsection{Inference using Metropolis-within-Gibbs}
\label{sec:algo}
\noindent
Equation~\ref{eqn:4} gives the posterior probability density of correlation matrix ${\bSigma}_C^{(S)}$, given
data ${\bf D}_S$. In our Metropolis-within-Gibbs based inference,
we update ${\bSigma}_C^{(S)}$--at which the partial correlation matrix $\bR$
is computed. Given this updated $\bR$, we then update the graph ${\mathbb
  G}(p, \bR)$. The graphical model comprising the credible-region defining set
of random Binomial graphs $\{{\mathbb G}(p, \bR)\}$ is thus learnt, where the
vertex set of each graph in this set is fixed as $\bV$; the ``credible
region'' in question is defined below in Section~\ref{sec:95}. 

In our learning of the $p\times p$-dimensional between-columns correlation
matrix ${\bSigma}_C^{(S)}$, the $\displaystyle{\frac{p^2-p}{2}}$ non-diagonal
elements of the upper (or lower) triangle are learnt, i.e. the parameters
$S_{12}, S_{13},\ldots, S_{1p}, S_{23},\ldots,S_{p-1\: p}$ are learnt. In the
$t$-th iteration of our inference, $S_{ij}$ is proposed from a Truncated
Normal density that is left truncated at -1 and right truncated at 1, as
$s_{ij}^{(t*)}\sim{{\cal TN}}(s_{ij}^{(t*)}; s_{ij}^{(t-1)}, v_{ij}, -1,
1),\quad,\forall\:i,j=1,\ldots,p;\:i\neq j,$ where $v_{ij}=v_0\forall\:i,j$ is
the experimentally chosen variance, and the proposal mean is the current value
$s_{ij}^{(t-1)}$ of $S_{ij}$ at the end of the $t-1$-th iteration. At the 2nd
block of the $t$-th iteration, the graph variable ${\mathbb G}(p, \bR)$ is
updated, given the current partial correlation matrix $\bR_t$, s.t. the
proposed edge variable connecting the $i$-th to the $j$-th vertex is
$g_{ij}^{(t\star)}\sim{Bernoulli}(g_{ij}^{(t\star)}; \rho_{ij}^{(t)})$, and
the $ij$-th proposed variance parameter is $\sigma_{ij}^{(t\star)}\sim{\cal
  N}(\sigma_{ij}^{(t\star)}; \sigma_{ij}^{(t-1)}, w_{ij}^2)$, where $w_{ij}^2$
are the experimentally chosen variance and the mean is the current value of
$\sigma_{ij}$. (Details in Section~1 of the Supplementary Material).

As suggested in Equation~\ref{eqn:4}, the correlation learning involves
computing 
$\left({\bSigma}_C^{(S)}\right)^{-1}$, $\vert {\bSigma}_C^{(S)}\vert$
and $\vert{\bf D}_S \left({\bSigma}_C^{(S)}\right)^{-1} \left({\bf
  D}_S\right)^T\vert$, in every iteration.
This calls for Cholesky decomposition of $\bSigma_C^{(S)}$ as
$\bL_C^{(S)}(\bL_C^{(S)})^T$, and of
${\bf D}_S \left({\bSigma}_C^{(S)}\right)^{-1} \left({\bf
  D}_S\right)^T$, into the (lower) triangular matrix $\bL$ and
$\bL^T$, while implementing ridge adjustment \ctp{wothke}. The latter computation follows the inversion of
$\bSigma_{C}^{(S)}$ into $(\bSigma_{C}^{(S)})^{-1}$, which is undertaken
using a forward substitution algorithm. (Details in Section~7 of the Supplementary Material).

\subsection{Defining the 95$\%$ HPD credible regions on the random graph variable, and the learnt graphical model}
\label{sec:95}
\noindent
We perform Bayesian inference on the random graph variable ${\mathbb
  G}(p,\bR)$, leading to one sampled graph at the end of each of the $N+1$
iterations of our inference scheme (Metropolis-within-Gibbs). In order to
acknowledge uncertainties in the Bayesian learning of the sought graphical
model, we need to include in its definition, only those graphs--sampled
post-burnin--that lie within an identified 95$\%$ HPD credible region.
We define the fraction $N_{ij}$ of the post-burnin
number $N_{post}$ of iterations (where $N_{post}< N+1$), in which the $ij$-th
edge exists, i.e. $G_{ij}$ takes the value 1,
$\forall\:i,j=1,2,\ldots,p,\:i\neq j$. Thus, variable
$N_{ij}$ takes the value  
\begin{equation}
n_{ij} := \displaystyle{\frac{\sum\limits_{t= N-N_{post}+1}^N g_{ij}^{(t)}}{N_{post}}},\quad i<j;\:i,j=1,\ldots,p,
\label{eqn:n}
\end{equation}
where the Bernoulli edge-variable $G_{ij}=g_{ij}^{(t)}$ in the $t$-th iteration.
Then $N_{ij}$ is the fractional number of sampled graphs, in which an edge
exists between vertices $i$ and $j$. This leads us to interpret
$\{N_{ij}\}_{i,j\in V;\:i< j}$ as carrying information about the uncertainty
in the graph learnt given data ${\bf D}_S$; in particular, $n_{ij}$
approximates the probability of existence of the edge between the $i$-th and
$j$-th nodes in the graphical model of the data at hand. Indeed the $N_{ij}$
parameters are functions of the partial correlation matrix $\bR$ that is
learnt given this data, but for the sake of notational brevity, we do not
include this explicit $\bR$ dependence in our notation to denote the edge
probability parameters.

So we view the set $\{{\mathbb G}(p,\bR_t)\}_{t=N-N_{post}+1}^N$ of graphs on
vertex set $\bV=\{1,\ldots,p\}$ and edge matrix $\bG_t$ in the $t$-th
iteration, that is updated given the current partial correlation matrix $\bR_t$ 
in the $t$-th iteration, equivalently as the post-burnin sample
\\$\{g^{(t)}_{12},g^{(t)}_{13},\ldots,g^{(t)}_{1p},g^{(t)}_{23},\ldots,g^{(t)}_{p\:p-1}\}_{t=N-N_{post}+1}^N$
of edge parameters. We include only those edge parameters in our defined
95$\%$ HPD credible region, that occur with probability $\geq 0.05$ in this
sample. In other words, only for $ij$ pairs s.t. $N_{ij}\geq 0.05$, define the
$g_{ij}$ parameters included in the set that comprises the 95$\%$ HPD credible
region on the edge parameters, in our definition. Indeed, the graphical model
of the data is then the set of those graphs on vertex set
$\bV=\{1,\ldots,p\}$, the existing edges of which are those $G_{ij}$
parameters that lie within this defined 95$\%$ HPD credible region.
\begin{definition}
\label{defn:95}
The graphical model of data ${\bf D}_S$ for which the between-column partial
correlation matrix is $\bR$, is the $\bR$-dependent set or family
$\displaystyle{{\cal G}_{p,{\bPhi}(\bR)}}$ of all inhomogeneous Binomial graphs ${\mathbb
  G}(p, \bR)$, the edge probabilities in which are given by the matrix $\bPhi(\bR)=[\phi_{ij}(R_{ij})]$, s.t. probability of the edge between the $i$-th and $j$-th nodes ($i\neq j;\:i,j\in V$) is 
\begin{equation}
\phi_{ij}(R_{ij}) = \left[H(n_{ij}- 0.05)\right] n_{ij}.
\label{eqn:graphmod}
\end{equation}
Here, $n_{ij}$ is the value of the parameter $N_{ij}$ defined in
Equation~\ref{eqn:n}, and $H(\cdot)$ is the Heaviside function \ctp{duffandnaylor66} where the Heaviside or step-function of $A\in{\mathbb R}$ is
\begin{eqnarray}
H(a) &=& 1 \quad {\mbox{if}}\quad a\geq 0\nonumber \\ 
     &=& 0 \quad {\mbox{if}}\quad a < 0.\nonumber 
\end{eqnarray}
Only edges with non-zero edge probability $\phi_{ij}(R_{ij})$, are marked on the learnt graphical model, and the corresponding
value of $N_{ij}$ is written next to each such marked edge. 
Then by this definition,
any graph ${\mathbb G}(p, \bR)\in\displaystyle{{\cal G}_{p,{\bPhi}(\bR)}}$ is sampled from within the 95$\%$ HPD credible region on inhomogeneous random Binomial graphs given the partial correlation matrix $\bR$ of the data.  
\end{definition}
Thus, in our approach, the binary edge parameter $G_{ij}$ between
the $i$-th and $j$-th nodes, takes the value 1 (i.e. the edge exists),
with a learnt probability--in fact, we learn the joint posterior of
all $G_{ij}$ parameters given the learnt correlation structure of the
data, while acknowledging the propagation of uncertainties in our
learning of the correlation given the data, into our learning of the
distribution of the $G_{ij}$ parameters given this learnt
partial correlation matrix $\bR$. A summary of this learnt distribution is then the edge probability parameter $\phi_{ij}(R_{ij})$,
the value of which is marked on the visualisation of the graphical
model of the data against the edge between the $i$-th and $j$-th
nodes, as long as $\phi_{ij}(R_{ij})>0$, i.e. $n_{ij}\geq 0.05$; $i\neq j;\:i,j\in \bV$. In other
words, only edges occurring with posterior probabilities in excess of 5$\%$ are
included in this graphical model. 
\section{Uncertainties in learnt graphical models help compute inter-graph distance}
\label{sec:hell}
\noindent
We compute the distance between the graphical models of two multivariate
datasets ${\bf D}_1$ and ${\bf D}_2$ of disparate sizes ($n_1$ and
$n_2$ respectively), to compute the
correlation between them; in
effect, the exercise can address the possible independence of the $pdf$s that
the two datasets are sampled from. This is of course a hard question
to address when the data comprise measurements of a high-dimensional
vector-valued observable. We compute the Hellinger distance
between the posterior probability density of the learnt graphical model
$\displaystyle{{\cal G}_{p,{\bPhi_1}(\bR_1)}}$ of data ${\bf D}_1$, the
between-columns partial correlation matrix of which is $\bR_1$,
and the posterior of the learnt graphical model $\displaystyle{{\cal
    G}_{p,{\bPhi_2}(\bR_2)}}$ given the other dataset. Here
$\bPhi_m(\bR_m)$ is the matrix, the $ij$-th element of which is the edge
probability $\phi_{ij}(R_{ij})=n_{ij}$ if $n_{ij}\geq 0.05$ and
$\phi_{ij}(R_{ij})=0$ if $n_{ij}< 0.05$. $i\neq j;\: i,j=1,\ldots,p_m;\:
m=1,2$. We need to consider the Hellinger distance between the posteriors of
the graphical models of two datasets with the same number of columns, as this
distance is defined between densities that share a common domain.

\begin{definition}
{Square of Hellinger distance between two probability density functions $g(\cdot)$ and $h(\cdot)$ over a common domain ${\cal X}\in{\mathbb R}^m$, with respect to a chosen measure, is 
\begin{eqnarray}
D_H^2(g,f) &=& \displaystyle{\int\left(\sqrt{g(\bx)} - \sqrt{h(\bx)}\right)^2 d\bx} \nonumber \\ 
&=& \displaystyle{\int g(\bx)d\bx + \int h(\bx)d\bx -2\int \sqrt{g(\bx)}\sqrt{h(\bx)}d\bx } \nonumber \\ 
&=& \displaystyle{2\left(1 -\int \sqrt{g(\bx)}\sqrt{h(\bx)}d\bx\right) }.
\label{eqn:hell}
\end{eqnarray}}
\end{definition}
The Hellinger distance is closely related to the Bhattacharyya distance
\ctp{bhat} between two densities:
$D_B(g,f) =
\displaystyle{-log\left[\int\left(\sqrt{g(\bx)}\sqrt{h(\bx)}\right)^2
    d\bx\right]}$.

From the joint posterior of all
edge and variance parameters given the partial correlation matrix $\bR_m$
(that is itself updated given the data ${\bf D}_S^{(m)}$), we marginalise the
$\sigma^2_{ij}$ parameters, $\forall i,j=1,\ldots,p,\:\: i\neq j$, to achieve the
joint posterior probability density of the graph edge parameters given the
partial correlation matrix of the data at hand. So, at the end of the $t$-th
iteration, we compute the value of posterior
$\pi(G^{(mt)}_{11},G^{(mt)}_{12},\ldots,G^{(mt)}_{p\;p-1}\vert \bR_{mt})$,
$t=0,\ldots,N_{iter}$. Given the availability of the posterior at discrete
points in its support, implementation of the integral in the definition of the
Hellinger distance is replaced by a sum. So for the $m$-th dataset, the
posterior of the graph edge parameters in the $t$-th iteration
$p_{m}^{(t)}:=\pi(G^{(mt)}_{11},G^{(mt)}_{12},\ldots,G^{(mt)}_{p\;p-1}\vert
\bR_{mt}),$ 
is employed to compute square of the (discretised version of the) Hellinger
distance between the two datasets as
\begin{equation}
D_H^2(p_{1},p_{2}) = \displaystyle{\frac{
\sum\limits_{t=N_{burnin}+1}^{N_{iter}}\left(\sqrt{p_{1}^{(t)}} - \sqrt{p_{2}^{(t)}}\right)^2}{N_{iter}-N_{burnin}}},
\label{eqn:sq_hell}
\end{equation}
The Bhattacharyya distance can be similarly discretised.

However, MCMC does not provide normalised posterior probability densities--as
we employ uniform priors on the variance parameters, the marginalised
posterior probability of the edge parameters is known only up to an unknown
scale. In fact, what we record at the end of the $t$-th iteration, is the
logarithm $\ln(p_{m}^{(t)})$ of the un-normalised posterior of the edges of
the graph given the $m$-th data ($m=1,\: 2$). Hence the Hellinger
distance between the 2 datasets that we compute is only known
upto a constant normalisation $S$ that we use to scale both $p_{1}^{(t)}$ and
$p_{2}^{(t)}$, $\forall \; t=0,\ldots, N_{iter}$. We choose this scale
parameter $S$, to ensure that the scaled, log posterior of the graph in the
$t$-th iteration, is easily exponentiable, as in
$\exp\left(\frac{\ln(p_{m}^{(t)})}{s}\right)$. One way of achieving this is to
choose the global scale $S$ as: 
\begin{equation}
s:={\max}\{(\ln(p_{1}^{(0)}),
\ln(p_{1}^{(1)}),\ldots,\ln(p_{1}^{(N_{iter})}),
\ln(p_{2}^{(0)}),\ldots,\ln(p_{2}^{(N_{iter})})\}.
\label{eqn:es}
\end{equation}

\begin{remark}
\label{rem:hell}
{
Squared Hellinger distance $D_H^2(p_{1},p_{2})$ between discretised posterior
probability densities of 2 graphical models, computed using
$\exp(\ln(p_m^{(t)})/s)$ in Equation~\ref{eqn:sq_hell}, is affected by 
scaling parameter $S$. This scale dependence is mitigated in our definition of 
the distance between 2 graphical models
as the difference between the ratio of this computed 
${D_H(p_{1},p_{2})}$, to the scaled uncertainty inherent in one graphical
model, and the ratio of ${D_H(p_{1},p_{2})}$, to the scaled uncertainty
in the other learnt graphical model.}
\end{remark} 
\begin{proposition}
\label{prop:1}
{
For correlation matrix $\bR_m$, and edge-probability matrix
$\bPhi_m(\bR_m)=[\phi_{ij}(R_{ij})]$ defined as in Equation~\ref{eqn:graphmod}, 
we define the graphical model $\displaystyle{{\cal G}_{p,{\bPhi_m}(\bR_m)}}$;
$m=1,2$, $i\neq j;\: i,j=1,\ldots,p_m$. 

The separation between two graphical models is
\begin{eqnarray}
\delta(\displaystyle{{\cal G}_{p,{\bPhi_1}(\bR_1)}}, \displaystyle{{\cal G}_{p,{\bPhi_2}(\bR_2)}})&:= &{\Big{\vert}}\sqrt{D_H^2(p_{1},p_{2})}/D_{max,s}(1)-\sqrt{D_H^2(p_{1},p_{2})}/D_{max,s}(2){\Big{\vert}}\nonumber\\
&= & D_H(p_{1},p_{2})\displaystyle{{\Big{\vert}}\frac{1}{D_{max,s}(1)}
  - \frac{1}{D_{max,s}(2)}{\Big{\vert}}},
\label{eqn:delta12}
\end{eqnarray}
 where
the Hellinger distance $D_H(p_{1},p_{2})$, between the 2 graphical models, is defined in Equation~\ref{eqn:sq_hell} and
\begin{eqnarray}
D_{max, s}(m) & := & {\max}\{\exp(\ln(p_{m}^{(0)})/s),
\exp(\ln(p_{m}^{(1)})/s),\ldots,\exp(\ln(p_m^{(N_{iter})})/s)\}-\nonumber \\
&&{\min}\{\exp(\ln(p_{m}^{(0)})/s),\exp(\ln(p_{m}^{(1)})/s),\ldots,\exp(\ln(p_m^{(N_{iter})})/s)\},
\label{eqn:graph_new}
\end{eqnarray}
computed for this chosen value $s$ of scale $S$ (defined in
Equation~\ref{eqn:es}), i.e. $D_{max,s}(m)$ provides separation between the
maximal and minimal (scaled values of) posteriors of graphs, generated in the
MCMC chain run using the $m$-th data; $m=1,2$.  }
\end{proposition}
Thus, the effect of the global scale is removed by comparing
$D_H(p_1,p_2)/D_{max,s}(1)$ to \\$D_H(p_1,p_2)/D_{max,s}(2)$, i.e. by computing
the ratio of the Hellinger distance between two graphical models, each of which
is normalised by its inherent uncertainty; (see connection to
Remark~\ref{rem:hell}).

Alternatively, we could define a (discretised version of the) odds ratio of unscaled logarithm of the unnormalised posterior densities of the graphical models learnt using MCMC, given the two datasets, as $\displaystyle{\int\left(\log(g(\bx))-{\log(h(\bx))}\right)d\bx }$; such is then a divergence measure that we define as
\begin{equation}
O_\pi(p_{1},p_{2}):= \displaystyle{
\sum\limits_{t=N_{burnin}+1}^{N_{iter}}
\left[{\log(p_{1}^{(t)}) -\log({p_{2}^{(t)}})}\right]
}.
\label{eqn:div}
\end{equation}

\subsection{Suggested inter-graph separation $\delta(\cdot,\cdot)$, is an inter-graph distance}
\label{sec:itisdist}
\begin{theorem}
{
Let $\delta(\displaystyle{{\cal G}_{p,{\bPhi_1}(\bR_1)}}, \displaystyle{{\cal
    G}_{p,{\bPhi_2}(\bR_2)}})$ be the separation 
between 2 with-uncertainty learnt graphical models defined over vertex set $\{1,\ldots,p\}$
($\displaystyle{{\cal G}_{p,{\bPhi_1}(\bR_1)}}$, and $\displaystyle{{\cal
    G}_{p,{\bPhi_2}(\bR_2)}}$, declared in Proposition~\ref{prop:1}), 
as defined in Equation~\ref{eqn:delta12}. Here the graphical model
$\displaystyle{{\cal G}_{p,{\bPhi_m}(\bR_m)}}$ is an element of space
${\bOmega}_p$, $m=1,2$.  

Then our definition of this inter-graph separation $\delta:{\bOmega}_p\times {\bOmega}_p\longrightarrow{\mathbb
  R}_{\geq 0}$, is a distance function, or a metric.}
\end{theorem} 

\begin{proof}
For $\delta:{\bOmega}_p\times {\bOmega}_p\longrightarrow{\mathbb
  R}_{\geq 0}$ to be a distance function or a metric, it should possess the
following properties.
\begin{enumerate}
\item $\delta(\displaystyle{{\cal G}_{p,1}}, \displaystyle{{\cal
      G}_{p,2}})\geq 0$ $\forall \displaystyle{{\cal G}_{p,1}}, \displaystyle{{\cal
    G}_{p,2}}\in{\bOmega}$, and $\delta(\displaystyle{{\cal G}_{p,1}},
\displaystyle{{\cal G}_{p,2}})=0 \iff \displaystyle{{\cal
    G}_{p,1}}=\displaystyle{{\cal G}_{p,2}}$.
\item $\delta(\displaystyle{{\cal G}_{p,1}}, \displaystyle{{\cal
      G}_{p,2}}) = \delta(\displaystyle{{\cal G}_{p,2}}, \displaystyle{{\cal
      G}_{p,1}})$ $\forall \displaystyle{{\cal G}_{p,1}}, \displaystyle{{\cal
    G}_{p,2}}\in{\bOmega}$
\item $\delta(\displaystyle{{\cal G}_{p,1}}, \displaystyle{{\cal
      G}_{p,3}}) \leq \delta(\displaystyle{{\cal G}_{p,1}}, \displaystyle{{\cal
      G}_{p,2}}) + \delta(\displaystyle{{\cal G}_{p,2}}, \displaystyle{{\cal
      G}_{p,3}})$, $\forall \displaystyle{{\cal G}_{p,1}}, \displaystyle{{\cal
    G}_{p,2}}, \displaystyle{{\cal G}_{p,3}} \in{\bOmega}$
\end{enumerate}
To abbreviate notation, we define: $$\ell_i := D_{max,s}(i),\quad,i=1,2,3.$$
Then we recall the definition of $\delta(\cdot,\cdot)$ as
$$\delta(\displaystyle{{\cal G}_{p,i}}, \displaystyle{{\cal G}_{p,j}}) := D_H(p_{i},p_{j})\displaystyle{{\Big{\vert}}\ell_i
  - \ell_j}{\Big{\vert}},$$ for datasets indexed by the integers $i$-th and
$j$. Below we consider 3 datasets indexed by $i=1,2,3$, the learnt graphical
models of which are $\displaystyle{{\cal G}_{p,i}}\in{\bOmega}$, the
separation between the maximal and minimal values of posterior probabilities
of which for a chosen global scale $S$ is $\ell_i:=D_{max,s}(i)$, and the
scaled, (by this $s$) discretised Hellinger distance between the posterior
probabilities of the graphical model $\displaystyle{{\cal G}_{p,i}}$ and
$\displaystyle{{\cal G}_{p,j}}$ is $D_H(p_{i},p_{j})$, $j=1,2,3$.

\noindent
--Proof of non-negativity:\\
in the definition of $\delta(\cdot,\cdot)$, $D_H(p_{1},p_{2})\geq 0$ is the Hellinger distance
between the posterior probability densities of the graphical models 
$\displaystyle{{\cal G}_{p,1}}, \displaystyle{{\cal G}_{p,2}}$. 
$\therefore \delta(\displaystyle{{\cal G}_{p,1}}, \displaystyle{{\cal
    G}_{p,2}}) \geq 0$.\\
Also, Hellinger distance between 2 probability densities, being a metric, is 0
$\iff$ the densities are equal. Then $\delta(\displaystyle{{\cal G}_{p,1}}, \displaystyle{{\cal
    G}_{p,2}})=0\Longrightarrow D_H(p_{1},p_{2})=0\Longleftrightarrow
\displaystyle{{\cal G}_{p,1}}=\displaystyle{{\cal G}_{p,2}}$.\\
As $D_{max,s}(\cdot)$ is probabilistically generated, we consider
$D_{max,s}(1)\neq D_{max,s}(2)$, for distinct posterior densities.

\noindent
--Proof of symmetry:\\
by definition, $\delta(\displaystyle{{\cal G}_{p,1}}, \displaystyle{{\cal
    G}_{p,2}})=\delta(\displaystyle{{\cal G}_{p,2}}, \displaystyle{{\cal
    G}_{p,1}})$, since $D_H(p_{1},p_{2})=D_H(p_{2},p_{1})$ by virtue of being
a metric, and $\displaystyle{{\Big{\vert}}\ell_1  - \ell_2{\Big{\vert}}}=
               \displaystyle{{\Big{\vert}}\ell_2  - \ell_1{\Big{\vert}}}$.

\noindent
--Proof of triangle-inequality obedience:\\
we aim to prove $$\delta(\displaystyle{{\cal G}_{p,1}}, \displaystyle{{\cal
    G}_{p,3}}) \leq \delta(\displaystyle{{\cal G}_{p,1}}, \displaystyle{{\cal
    G}_{p,2}}) + \delta(\displaystyle{{\cal G}_{p,2}}, \displaystyle{{\cal
    G}_{p,3}}),\:\:{\mbox{i.e.}}$$ 
$$D_H(p_{1},p_{3})\vert \ell_1 - \ell_3\vert \leq D_H(p_{1},p_{2})\vert \ell_1 - \ell_2\vert+ D_H(p_{2},p_{3})\vert \ell_2 - \ell_3\vert,$$
given
\begin{equation}
D_H(p_{1},p_{3}) \leq D_H(p_{1},p_{2})+ D_H(p_{2},p_{3}),
\label{eqn:helltri}
\end{equation}
(the Hellinger distance being a metric obeys the triangle inequality).

\noindent
We assume:
$$\delta(\displaystyle{{\cal G}_{p,1}}, \displaystyle{{\cal
    G}_{p,3}}) > \delta(\displaystyle{{\cal G}_{p,1}}, \displaystyle{{\cal
    G}_{p,2}}) + \delta(\displaystyle{{\cal G}_{p,2}}, \displaystyle{{\cal
    G}_{p,3}}),\quad {\mbox{i.e.}}$$
$$D_H(p_{1},p_{3})\vert \ell_1 - \ell_3\vert > 
D_H(p_{1},p_{2})\vert \ell_1 - \ell_2\vert +
D_H(p_{2},p_{3})\vert \ell_2 - \ell_3\vert \quad {\mbox{}}$$
Then this equation, together with inequation~\ref{eqn:helltri}, tells us
\begin{eqnarray}
D_H(p_{1},p_{2})\displaystyle{{\vert \ell_1 - \ell_2\vert}} +
D_H(p_{2},p_{3})\displaystyle{{\vert \ell_2 - \ell_3\vert}} &<& 
D_H(p_{1},p_{3}){\vert \ell_1 - \ell_3\vert}\nonumber \\
 \leq D_H(p_{1},p_{2}){\vert \ell_1 - \ell_3\vert} &+& D_H(p_{2},p_{3}){\vert \ell_1 - \ell_3\vert}  \nonumber 
\end{eqnarray}
i.e.
\begin{eqnarray}
D_H(p_{1},p_{2})\displaystyle{{\vert \ell_1 - \ell_2\vert}} +
D_H(p_{2},p_{3})\displaystyle{{\vert \ell_2 - \ell_3\vert}} &<& \nonumber \\
D_H(p_{1},p_{2}){\vert \ell_1 - \ell_3\vert} + D_H(p_{2},p_{3}){\vert
   \ell_1 - \ell_3\vert} &&  
\label{eqn:realtri} 
\end{eqnarray}
Now let $\ell_1=\ell_3$, which we consider to occur only if the graphical
model due to the dataset with index 1, equals the graphical model model due to
dataset with index 3, i.e. if datasets with indices 1 and 3 are the same. In
this case, $D_H(p_{1},p_{3})=0$, but by inequation~\ref{eqn:helltri}, 
$D_H(p_{1},p_{2})$ and $D_H(p_{2},p_{3})$ are not necessarily 0.
The RHS of inequation~\ref{eqn:realtri}
is then 0, but the LHS is not negative, i.e. the case $\ell_1=\ell_3$ is a
counterexample against the validity of inequation~\ref{eqn:realtri}. Thus,
inequation~\ref{eqn:realtri} is false $\Longrightarrow$our assumption is false. Therefore,
$$\delta(\displaystyle{{\cal G}_{p,1}}, \displaystyle{{\cal G}_{p,3}}) \leq
\delta(\displaystyle{{\cal G}_{p,1}}, \displaystyle{{\cal G}_{p,2}}) + \delta(\displaystyle{{\cal G}_{p,2}}, \displaystyle{{\cal G}_{p,3}}).$$
This proves that $\delta(\cdot,\cdot)$ abides by the triangle inequality.
Thus the inter-graph separation $\delta(\cdot,\cdot)$ that
we introduced in Proposition~\ref{prop:1}, on learnt graphical models that
live in space $\bOmega_p$, is a metric or a distance function, that gives 
the inter-graph distance.
\end{proof}

\begin{proposition}
  \label{prop:2} { For a given value of the inter-graph distance
    $\delta(\displaystyle{{\cal G}_{p,1}}, \displaystyle{{\cal
        G}_{p,2}})\in[0,\infty)$, between 2 learnt graphical models
    $\displaystyle{{\cal G}_{p,2}}\displaystyle{{\cal
        G}_{p,1}}\in{\bOmega}_p$, defined over vertex set $\{1,\ldots,p\}$,
    where the graphical model $\displaystyle{{\cal G}_{p,\cdot}}$ is learnt given data
    $\bD_\cdot$, a 
    model for the absolute value of the correlation $\vert corr(\bZ_1,\bZ_2)\vert $ between the $p$-dimensional vector-valued observable
    $\bZ_1$, ($n_1$ measurements of which comprise dataset indexed by 1), and the
    $p$-dimensional observable $\bZ_2$, ($n_2$ measurements of which comprise
    dataset indexed by 2), is
$$\delta(\displaystyle{{\cal G}_{p,1}}, \displaystyle{{\cal
        G}_{p,2}}) = -\log\left(\vert corr(\bZ_1,\bZ_2)\vert\right),$$   
$$s.t.\:\: \vert corr(\bZ_1,\bZ_2)\vert = \exp[-\delta(\displaystyle{{\cal G}_{p,1}},
\displaystyle{{\cal G}_{p,2}})]\in(0,1].$$ 
}
\end{proposition}

\section{Implementation on real data}
\label{sec:real}
\noindent
\begin{sloppypar}
{In this section we make applications of our method to the relatively
well-known data sets on 11 different chemical attributes and ``quality''
classes of red and white wines, grown in the Minho region of Portugal
(referred to a ``vinho verde''); these data have been considered by
\ctn{CorCer09} and discussed in
\url{https://onlinecourses.science.psu.edu/stat857/node/223} (hereon PSU). The data
consists of information on 1599 red wines and 4898 white wines. Each of these
data sets consists of 12 columns that contain information on vino-chemical
attributes of the sampled wines; these properties are assigned the following
names: ``fixed acidity'' ($X_1$), ``volatile acidity'' ($X_2$), ``citric
acid'' ($X_3$), ``residual sugar'' ($X_4$), ``chlorides'' ($X_5$), ``free
sulphur dioxide'' ($X_6$), ``total sulphur dioxide'' ($X_7$), ``density''
($X_8$), ``pH'' ($X_9$), ``sulphates'' ($X_{10}$), ``alcohol'' ($X_{11}$) and
``quality'' ($X_{12}$). Then the $n$-th row and $i$-th column of the data
matrix carries measured/assigned value of the $i$-th property of the $n$-th
wine in the sample, where $i=1,\ldots,12$ and $n=1,\ldots,n_{orig}=1599$ for
the red wine data ${\bf D}^{(red)}_{orig}$, while $n=1,\ldots,n_{orig}=4898$
for the white wine data ${\bf D}^{(white)}_{orig}$. We refer to the $i$-th
vinous property to be $X_i$. Then $X_i\in{\mathbb R}_{\geq 0}$ $\forall
i=1,\ldots,11$, while $X_{12}$ that denotes the perceived ``quality'' of the
wine is a categorical variable. Each wine in these samples was assessed by at
least three experts who graded the wine on a categorical scale of 0 to 10, in
increasing order of excellence. The resulting ``sensory score'' or value of
the ``quality'' parameter was a median of the expert assessments
\ctp{CorCer09}. We seek the graphical model given each of the wine data sets,
in which the relationship between any $X_i$ and $X_j$ is embodied, $i\neq j;\;
i,j=1,\ldots, 12$. Thus, we seek to find out how the different vino-chemical
attributes affect each other, as well as the quality of the wine, in the
sample at hand. Here, $X_1,\ldots,X_{11}$ are real-valued, while $X_{12}$ is a categorical
variable, and our methodology allows for the learning of the graphical model
of a data set that in its raw state bears measurements of variables of
different types. In fact, we standardise our data, s.t. $X_i$ is standardised
to $Z_i$, $i=1,\ldots,p$, $p=12$. We work with only a subset data set,
(comprising only $n<n_{orig}$ rows of the available ${\bf
  D}^{(\cdot)}_{orig}$; $n=300$ typically). Thus, the data sets with $n$ rows,
containing $Z_i$ values, ($i=1,\ldots,p=12$), are $n\times p$-dimensional
matrices each; we refer to these data sets that we work with, as ${\bf
  D}^{(white)}_S$ and ${\bf D}^{(red)}_S$, respectively for the white and red
wines. Our aim is to learn the between-column correlation matrix
$\bSigma_S^{(m)}$ given data ${\bf D}^{(m)}_S$, and simultaneously
learn the graphical model of this data using the methodology that we
have developed above; $m=white,\:red$.}
\end{sloppypar}

The motivation behind choosing these data sets are basically
three-fold. Firstly, we sought multivariate, rectangularly-shaped, real-life data,
that would admit graphical modelling of the correlations between the different
variables in the data. Also, we wanted to work with data, results from--at
least a part of--which exists in the literature. Comparison of these published
results, with our independent results can then illustrate strengths of our
method. Thirdly, treating the red and white wine data as data realised at
different experimental conditions, we would want to address the question of
the distance between these data, and we propose to do this by computing the
distance between the graphical models of the two data sets. Hence our choice
of the popular Portuguese red and white wine data sets, as the data that we
implement to illustrate our method on. It is to be noted that a rigorous
vinaceous implications of the results, is outside the scope and intent of this
paper. However, we will make a comparison of our results with the results of
the analysis of white wine data that is reported in PSU
precludes analysis of the red wine data.

\subsection{Results given data ${\bf D}^{(white)}_S$}
\noindent
\begin{figure}[!t]
\centering
\includegraphics[width=10cm,height=8cm]{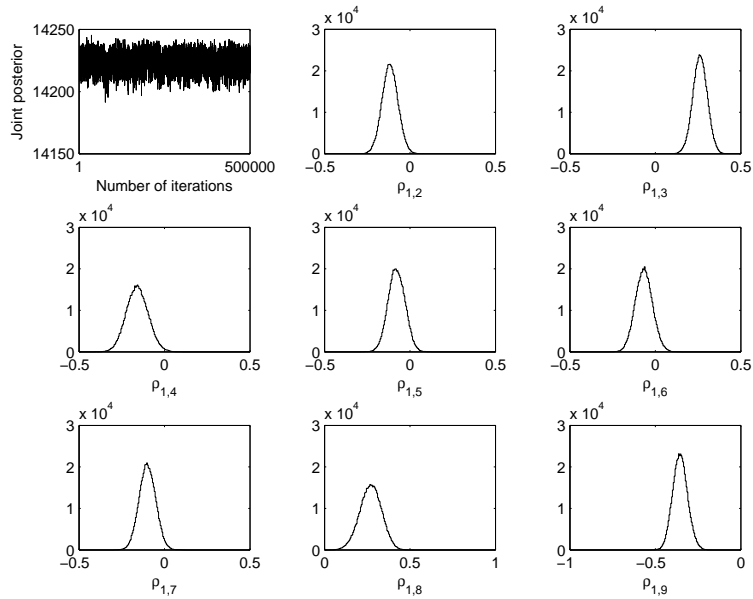}
\vspace{.0cm}
\caption{{\it Top left panel:} trace of the joint posterior
  probability density of the elements of the upper triangle of the
  between-columns correlation matrix of the standardised version of
  the real data ${\bf D}_S^{(white)}$ on Portuguese white wine samples
  \ctp{CorCer09}; this data has $n=300$ rows nd $p=12$ columns, and is
  constructed as a randomly sampled subset of the original data, the
  sample size of which is 4898. {\it All other panels:} histogram representations of marginal posterior probability densities of some of the partial correlation parameters computed using the correlation matrix learnt given data ${\bf D}_S^{(white)}$.
} 
\label{fig:white_corr}
\end{figure}

\begin{figure}[!t]
{
\hspace*{-1in}
\includegraphics[width=15cm,height=6cm]{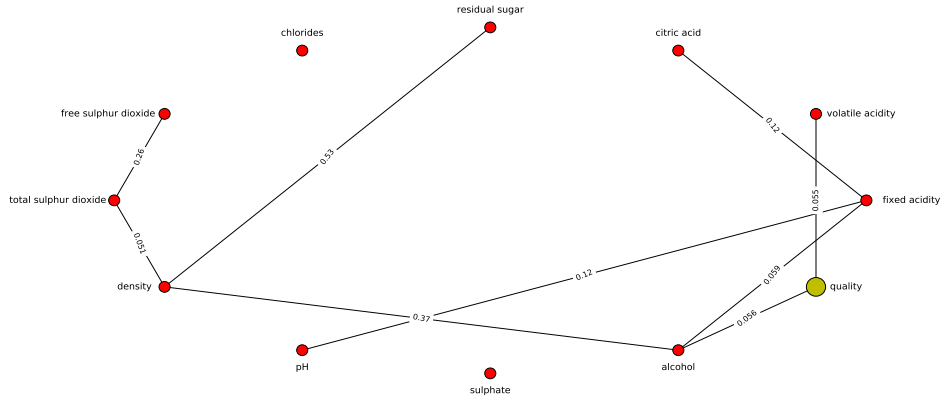}
}
\vspace*{-.2in}
\caption{Figure showing graphical model of standardised version ${\bf
    D}_S^{(white)}$, of the real data on Portuguese white wine samples
  \ctp{CorCer09}. Each of the first 11 columns of this data gives the measured
  value of each of 11 different vino-chemical properties of the wines in the
  sample--marked as nodes in the graph above, by filled red (or grey in the
  printed version) circles, with the name of the property included in the
  vicinity of the respective node. The 12-th column in the data includes
  values of the assessed quality of a wine in the sample, (a node that we mark
  with a green circle in the electronic version; the bigger grey circle in a
  monochromatic version of the paper). The probability for an edge to exist in
  the post-burnin sample of graphs generated in our MCMC-based inferential
  scheme, is marked against an existing edge, where edges with such
  probabilities that are $< 0.05$ are omitted from this graphical model, as
  included within a pre-defined 95$\%$ HPD credible region (defined in
  Section~\ref{sec:95}) on the MCMC-based sample of graphs.}
\label{fig:white_graph}
\vspace*{-.1in}
\end{figure}

The top left-hand panel of Figure~\ref{fig:white_corr} presents the trace of
the joint posterior probability density of the correlation parameters $S_{ij}$
of the upper triangle of the
between-column correlation matrix $\bSigma_S^{(white)}$, given the
standardised white wine data ${\bf D}^{(white)}_S$ that we choose to consist of $n=300$ number of rows and $p=12$ number of
columns. All the other panels of this figure include marginal posterior probabilities of some of the partial correlation parameters, with value $\rho_{ij}$, where the $i$-th variable is the $i$-th vinous parameter listed above, with $i=1,\ldots,12;\: j\neq i, j=1,\ldots,12$.
Figure~9 in Supplementary Materials presents trace of the
joint posterior of the $G_{ij}$ and $\sigma_{ij}^2$ parameters,
updated in the 2nd block of each iteration of our MCMC chain, at the
updated (partial) correlation matrix. 
Thus we obtain the sample of graphs,
$\{{\mathbb G}^{(t)}(p,\bR_t)\}_{t=N-N_{post}+1}^N$, where each graph
is on the vertex set $\bV=\{1,\ldots,p\}$ and is learnt given the
partial correlation matrix $\bR_t$ in the $t$-th iteration of our MCMC
chain. We compute the graph edge probability parameter
$\phi_{ij}(R_{ij})$ for each $ij$-pair of nodes in this sample, and
include only those edges in the graphical model of the ${\bf
  D}^{(white)}_S$ data, that have non-zero $\phi_{ij}(R_{ij})$,
i.e. $n_{ij}\geq 0.05$ (see Section~\ref{sec:95}). For these edges,
the value $n_{ij}$ is marked against the edge between the $i$-th and
$j$-th nodes in the representation of this graphical model of this
white wine data set, that is shown in
Figure~\ref{fig:white_graph}. Here $i\neq j,\:i,j=1,\ldots,p=12$.
\subsubsection{Comparing against earlier work done with white wine data}
\noindent   
Comparison of our results with previous work done with the white wine data is
discussed in Section~4 of the Supplementary Section. Such previous work
includes ``Exploratory Data Analysis'' reported in
PSU using the white
wine data. In this work, a
matrix of scatterplots of $X_i$ against $X_j$, is presented; $i\neq j;\:
i,j=1,\ldots,11$. These empirical scatterplots
visually suggest stronger correlations between fixed acidity and pH;
residual sugar and density; free sulphur dioxide and total sulphur dioxide;
density and total sulphur dioxide; density and alcohol--than amongst other
pairs of variables. These are the very node pairs that we identify to have
edges (at probability in excess of 0.05) between them. Existence of edges to/from the ``quality''
variable, is corroborated by examining the results reported in that work, on regressing
this variable against the others. This regression analysis of the predictors $X_1,\ldots,X_{11}$ on the
response variable ``quality'' suggests the variables
alcohol and volatile acidity to have maximal effect on quality. 
Indeed, this is corroborated in our learning
of the graphical model that manifests edges between the nodes corresponding to
variables: alcohol-quality, and volatile acidity-quality.

\begin{figure}[!t]
{
\hspace*{-1in}
\includegraphics[width=15cm,height=6cm]{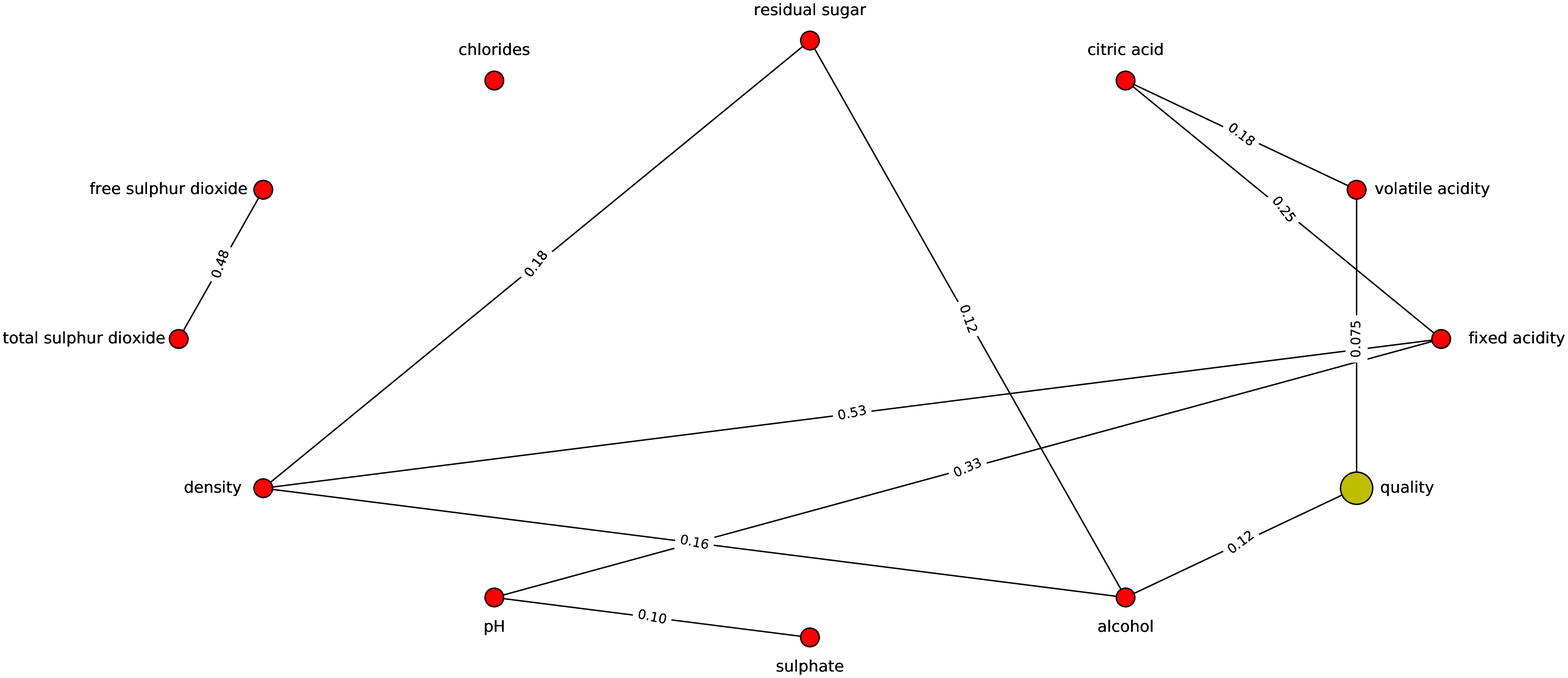}
}
\vspace*{-.2in}
\caption{Graphical model of standardised version ${\bf D}_S^{(red)}$ of the
  real data on Portuguese red wine samples
  \ctp{CorCer09}. Figure is similar to Figure~\ref{fig:white_graph}, except that this is the graphical model learnt for the red wine data.}
\label{fig:red_graph}
\vspace*{-.2in}
\end{figure}

\subsection{Results given data ${\bf D}^{(red)}_S$}
\label{sec:red}
\noindent
The ${\bf D}^{(red)}_S$ data is the standardised version of a subset of the
original red wine data set ${\bf D}^{(red)}_{orig}$. ${\bf D}^{(red)}_S$
comprises $n=300$ rows and $p=12$. 
The marginal posterior of some of the partial correlation parameters
$\rho_{ij}$ computed using the elements of the correlation matrix
$\bSigma^{(red)}_S$ (of data ${\bf D}^{(red)}_S$) that is updated in the first
block of Metropolis-within-Gibbs, are presented in Figure~10 of the
Supplementary Section.
In the second block, we update the edge parameters $G_{ij}$ of the graph
${\mathbb G}(p,\bR)$ given the newly updated partial corelation matrix $\bR$.
Figure~11 of the Supplementary Section presents the trace of the joint
posterior probability of the $G_{ij}$ parameters and the variance parameters
$\sigma_{ij}^2$ (of the Normal likelihood; see Equation~\ref{eqn:folded}), given data ${\bf
  D}^{(red)}_S$. The marginal of some of the variance parameters are also
shown in the other panels of this figure. The inferred graphical model of the
red wine data is included in Figure~\ref{fig:red_graph}.
\subsubsection{Comparing against empirical work done with red wine data}
\noindent   
To the best of our knowledge, analysis of the red wine data has not been
reported in the literature. In lieu of that, we undertake an empirical and
regression analysis of this red wine data, and compare our learnt results with
results of such analyses in Section~6 of the Supplementary Material. 
We further undertook a modelling of the relationship between the response
variable ``quality'' ($Z_{12}$) and the other 11 covariates ($Z_1$ to
$Z_{11}$), via an OLS regression in which quality is
regressed over the other vino-chemical attributes). This modelling suggests
the strongest effect of alcohol and volatile-acidity on quality (see
Figure~14 of Supplementary Material); this trend is replicated in our
learnt graphical model of the red wine data.

\section{Metric measuring distance between posterior probability densities of graphs given white and red wine datasets}
\label{sec:real_hell}
\noindent
We seek the distance $\delta(\cdot,\cdot)$ that we defined in
Proposition~\ref{prop:1}, between the learnt
red and white wine graphs, using the method delineated in
Section~\ref{sec:hell}. For this, we first compute 
the normalisation\\ $S:={\max}\{(\ln(p_{red}^{(0)}),
\ln(p_{red}^{(1)}),\ldots,\ln(p_{red}^{(N_{iter})}),
\ln(p_{white}^{(0)}),\ldots,\ln(p_{red}^{(N_{iter})})\}$, which for the red
and white wine datasets yields
$s=\ln(p_{red}^{(1474)})\approx 142.7687$. We then use $\exp(\ln(p_m^{(t)})/s)$ in Equation~\ref{eqn:sq_hell}; $m=white, red$. 
Then scaling the log posterior given either data set, at
any iteration, by the global scale value of $s$=142.7687 approximately, we get
$D_H(p_{white},p_{red})\approx 0.1153$, so that the logarithm of this
value of the Hellinger distance between the 2 learnt graphical models is $\ln(0.1153)\approx -2.1602$.
Similarly, using the same scale, the Bhattacharyya distance is
$D_B(p_{white},p_{red})\approx -1.7623$, where we recall that this
measure is a logarithm of the distance.

\begin{figure}[!hbt]
{
\hspace*{0in}
\includegraphics[width=12cm]{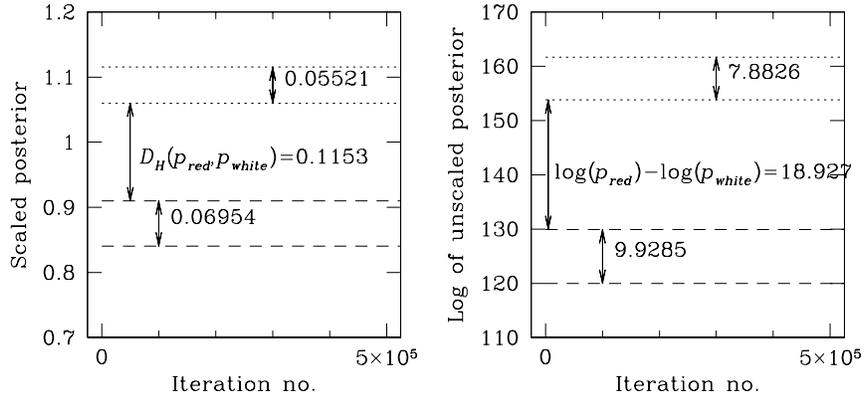}
}
\vspace*{-.0in}
\caption{{\it Left}: minimum and maximum values of the scaled posterior
  probability density of the graph sampled in an iteration in the MCMC chain
  run with the red wine data, plotted in dotted lines against the number of
  the iteration. The difference between these values is depicted within the
  band delineated by these lines. The broken lines show the same for the
  results obtained from the MCMC chain run using the white wine data. The
  value of the Hellinger distance $D_H(p_{red},p_{white})$ computed using the
  scaled posterior probabilities of the graphical models given the two wine
  data sets, is also marked, as about 0.1153. All log posterior values are
  scaled by a chosen global scale and exponentiated (as discussed in the
  text). {\it Right}: similar to the left panel, except that here, the ratio of
  the logarithm of the unscaled posteriors is used; the value of the log odds
  between the posteriors of the red and white wine data sets is marked to be
  about 18.927.}
\label{fig:dist}
\end{figure}
For this $s$ and the red wine data, we compute the uncertainty inherent in
graphical model of the red-wine data as $D_{max,s}(red)$,
between the graph that occurs at maximal posterior and that at the minimal
posterior (Equation~\ref{eqn:graph_new}). Similarly, we compute
$D_{max,s}(white)$. We then compute ratio of the Hellinger distance between the
graphical models learnt given the red and white-wine data, to the uncertainty
inherent in each learnt model, and compare
${D_H(p_{white},p_{red})}/D_{max,s}(red)$, with
${D_H(p_{white},p_{red})}/D_{max,s}(white)$. This comparison is
depicted in the left panel of Figure~\ref{fig:dist} that shows that the
difference $D_{max,s}(white)$ between the scaled posterior of graphs given the
white wine data is about 0.0694 while $D_{max,s}(red)$ given the red wine data
is about 0.05521, These values are compared to the Hellinger distance (between
scaled posteriors) of about 0.1153, between graphs given the red and white
wine data.  Thus, $D_H(p_{red},p_{white})$ is about 1.66$D_{max,s}(white)$ and
about 2.1$D_{max,s}(red)$. 
Thus, our inter-graph distance metric, between the graphical models learnt
given the two data sets is $$\delta(white, red)\approx 0.44$$. 
Then intuitively speaking, this inter-graph distance between the graphical models given the red and white wine
datasets, may suggest independence of the data sets. 
Again, using the correlation model suggested in Proposition~\ref{prop:2}, the
absolute value of the correlation between the 12-dimensional vino-chemical
vector-valued measurable for the red wine data and that for the white wine
data, is $$\vert corr(white, red)\vert := \exp[-\delta(white, red)]\approx
0.1030,$$ which is a low correlation, indicating that the two graphical models
learnt given the real red and white wine Portuguese datasets, are not sampled
from the same $pdf$.

Compared to these, the sample mean of the log odds of the posterior of the
graphs generated in the post-burnin iterations, given the two data is 18.9273,
which is about 1.9 times the maximal difference between the log posterior
values of graphs achieved in the MCMC run with the white wine data, and about
2.4 times that for the red wine data (see Figure~\ref{fig:dist}). Again,
this suggests that the log odds as a measure of divergence between the
graphical models given these two wine data sets, is significantly higher than
the uncertainty internal to the results for each data.

This clarifies how our pursuit of uncertainties in learnt graphical models,
and inter-graph distance, share an integrated umbrage of purpose, where the
former leads to the latter.

\section{Learning the human disease-symptom network}
\label{sec:disease}
\noindent
Our methodology for learning the graphical model, can be
implemented even for a highly multivariate data that generates a
graph with a very large number of nodes. In this section, we discuss such a
graph (with $\gtrsim$8000 nodes) that describes the correlation structure of
the human disease-symptom network. 
\begin{figure}[!hbt]
\hspace*{-3.0cm}{
\includegraphics[width=18cm,height=12cm]{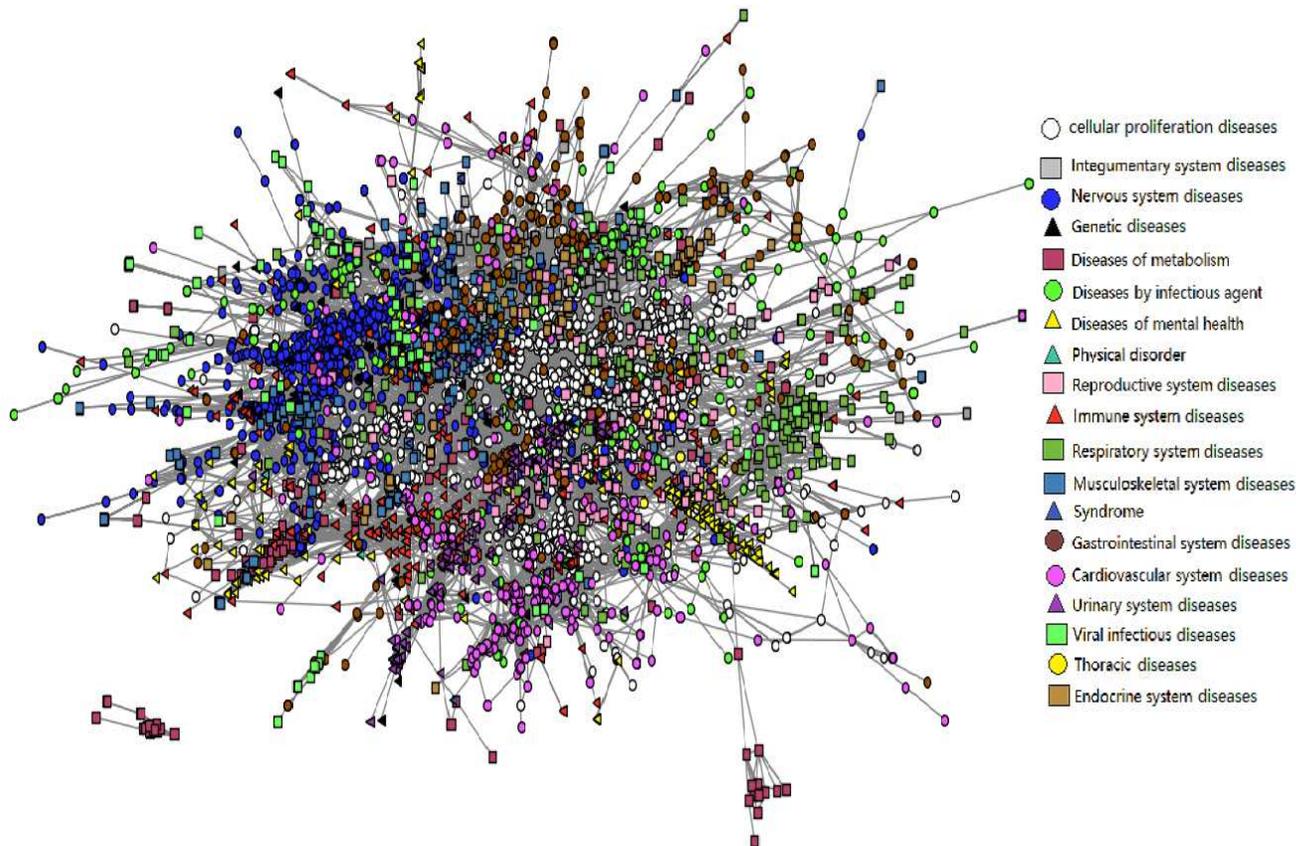}}
\vspace{-.5cm}
\caption{The human disease phenotype graphical model that we learn using the
  disease-disease partial correlation obtained using the computed Spearman
  rank correlation between the rank vectors of a list of phenotypes, where the
  phenotype ranking reflects semantic relevance of a phenotype to the disease
  in question (quantified by HSG as the NPMI parameter in the ${\bf D}_{DPh}$
  dataset). Only edges with posterior probability $\geq 0.9$ are included in
  this graph, and nodes that have edges with posterior less than 0.9, are
  discarded, resulting in 6052 diseases (nodes) remaining in this graph. There
  are 145210 edges in the displayed graph. All diseases identified by name by
  HSG, to belong to one of the 19 given disease class, are presented above in the same colour; the colour key identifying these classes, is attached. To draw the graph, we used a Python-based code that implements the Fruchterman-Reingold force-directed algorithm.}
\label{fig:disgraph}
\end{figure}

\ctn{hsg} (HSG hereon) learn this network by considering the similarity
parameter for each pair of diseases that are elements of an identified set of
diseases in the Human Disease Ontology (DO), that contains
information about rare and common diseases, and spans heritable,
developmental, infectious and environmental diseases.  Here, the ``similarity
parameter'' between one disease and another, is computed using the ranked
vectors of ``normalised pointwise mutual information'' (NMPI) parameters for the two diseases, where the NMPI
parameter describes the relevance of a symptom (or rather, a phenotype), to
the disease in question. HSG define the NMPI parameter semantically, as the
normalised number of co-occurrences of a given phenotype and a disease in the
titles and abstracts of 5 million articles in Medline. To do this, they make
use of the Aber-OWL: Pubmed infrastructure that performs such semantical
mining of the Medline abstracts and titles. The disease-disease pairwise
semantic similarity parameters--computed using the degree of overlap in the
relevance ranks of phenotypes associated
with each disease--result in a similarity matrix, which HSG turn into a
disease–-disease network based on phenotypes. To do this, they only choose from
the top-ranking 0.5$\%$ of disease–-disease similarity values. Phenotypes
associated with diseases, and corresponding scoring functions (such as
the NPMI), exist in the file ``doid2hpo-fulltext.txt.gz'' at
\url{http://aber-owl.net/aber-owl/diseasephenotypes}. In fact, 
this file
contains information about $N_{dis}$ diseases, and
the semantic relevance of each of the $N_{pheno}$ phenotypes to each disease,
as quantified by NPMI parameter values, in addition to other scores such as
$t$-scores and $z$-scores. In this file, $N_{dis}$ is 8676 and $N_{pheno}$ is
19323. In the phenotypic similarity network between diseases that HSG report,
diseases are the nodes, and the edge between two nodes exists in this
undirected graph, if the similarity between the nodes (diseases) is in the
highest-ranking 0.5$\%$ of the 38,688,400 similarity values. They remove all
self-loops and nodes with a degree of 0. Their network is
presented in \url{http://aber-owl.net/aber-owl/diseasephenotypes/network/}.
The network analysis was performed using standard softwares and they identify
multiple clusters in their network, with agglomerates of some clusters (of
diseases), found to correspond to known disease-classes. The ``Group
Selector'' function on their visualisation kit, allows for the identification
of 19 such clusters in their disease-disease network, with each cluster
corresponding to a disease-class. 
The sum of the number of nodes over their
identified 19 clusters, is 5059. The number of edges in their network is
reported to be 65,795. The average node degree is then about 26.2. We discuss
detailed comparison of our results to HSG's in the following subsection,
including comparison of HSG's and our recovery of
the relative number of nodes i.e. diseases, in each of the 19
disease classes that HSG classify their reported network into, and our
computed ratios of the averaged intra-class to inter-class variance for     
each of the 19 classes, compared to the ROC Area Under Curve       
values reported by HSG for each class.

HSG's network then manifests a similarity-structure that is computed using 
available NPMI parameter values. Our interest is in learning the
disease-disease graphical model, with each edge of such a graphical model
learnt to exist at a learnt probability. We perform such learning using the
NPMI semantic-relevance data that is made available for each of the $N_{dis}$
number of diseases, by HSG--we refer to this data as the human
disease-phenotype data ${\bf D}_{DPh}$. Using ${\bf D}_{DPh}$, we first
compute the partial correlation between any pair of diseases, for each of
which, information on the ranked (semantic) relevance of each of the
$N_{pheno}$ phenotypes exist, in this given dataset. Upon computation of
pairwise partial correlations, the graphical model for the ${\bf D}_{DPh}$
data is learnt.

We compute the partial correlation $R_{ij}$ between the $i$-th and $j$-th
diseases in the ${\bf D}_{DPh}$ data, ($i,j=1,\ldots,N_{dis}$, $i\neq j$), in
the following way. We rank the NPMI parameter values for the $i$-th
disease and each of the $N_{pheno}$ phenotypes, with the phenotype of the 
highest semantic relevance to the $i$-th disease assigned a rank 1. Let the 
rank vector of phenotypes, by semantic relevance to the $i$-th disease take the value ${\bf{{\mathpzc{r}}_i}}$
and similarly, the rank vector of phenotypes relevant to the $j$-th disease
is ${\bf{{\mathpzc{r}}_j}}$. We compute the Spearman rank correlation ${\mathpzc{s}}_{ij}^{(rank)}$,
of vectors ${\bf{{\mathpzc{r}}_i}}$ and ${\bf{{\mathpzc{r}}_j}}$. Then we
compute the partial correlation $R_{ij}$ $\forall\:
i,j=1,\ldots,N_{dis};\:i\neq j$, between
the $i$-th and $j$-th nodes of our undirected graph, using the computed
values of the Spearman rank correlation in $\{{\mathpzc{s}}_{ij}^{(rank)}\}$. It is useful to define the partial correlation using the Spearman rank correlation, rather than the correlation between the vector of normalised NPMI values, since we intend to correlate the $i$-th disease with the $j$-th disease depending on how relevant a given list of phenotypes is, to each disease, i.e. depending on the ranked relevance of the phenotypes. 

To learn the graphical model given this partial correlation structure in
$\bR=[R_{ij}]$ (that is itself computed from the data ${\bf D}_{DPh}$), in the
previous sections, we have delineated an MCMC-based inference strategy, that
helps us learn the edge parameters, as well as the variance of the
likelihood. However, the data that we want to learn the graphical model for,
is so highly multivariate--i.e. there are so many edges in the proposed
graph--that we forego iterating over the multiple samples of edge and
variance parameter values, and compute the graphical model for this data, by
computing the posterior probability for each edge, given the computed partial
correlation structure. In fact, the graphical model of data ${\bf D}_{DPh}$ that we present, comprises only those edge parameters, the posterior
probability of which exceeds 0.9.

Here, the posterior probability density of the edge $G_{ij}$ (=0 or 1) between
the $i$-th and $j$-th diseases, is proportional to the likelihood and prior:
$$\pi(G_{ij}\vert R_{ij}) \propto \ell(G_{ij}\vert R_{ij})\pi_0(G_{ij}),$$
where the prior on $G_{ij}$ is $Bernoulli(0.5)$ $\forall i,j$, and the
likelihood is the Normal likelihood that we chose to work with in our learning, as discussed before in Section~\ref{sec:graph}, i.e.
likelihood given $\bR=[R_{ij}]$ is  
$$
\displaystyle{
\prod\limits_{i\neq j; i,j=1}^{N_{dis}} \frac{1}{\sqrt{2\pi}\sigma_{ij}}
\exp\left[-\frac{\left(G_{ij} - R_{ij}\right)^2}{2\sigma_{ij}^2} 
-\frac{\left(G_{ij} + R_{ij}\right)^2}{2\sigma_{ij}^2}\right]
},$$
where the variance parameters $\{\sigma_{ij}\}_{i\neq j;i,j=1}^p$ are
defined as $\sigma_{ij}^2 = R_{ij}(1- R_{ij})$.

\begin{definition}
  Our visualised graph is a sub-graph of the full graph
  ${\mathbb G}(N_{dis}, \bR)$ of data ${\bf D}_{DPh}$, the between-columns
  partial correlation matrix of which is $\bR=[R_{ij}], \:i\neq j, \;
  i,j=1,\ldots,N_{dis}$, such that this visualised graph is defined to
  consist only of edges in the set: $\bE^{/}:=\{G_{ij}=1\vert
  \pi(G_{ij}\vert R_{ij})\geq 0.9; \:i\neq j, \;
  i,j=1,\ldots,N_{dis}\}$. This visualised graph has 6052 number of nodes
  (diseases) and 145210 edges, so that the average node degree is
  about 24. It is a random undirected graphical model and represents our learning of the human disease
  phenotype graph (displayed in Figure~\ref{fig:disgraph}).
\end{definition}  

\subsection{Comparing our results to the earlier work done on the human disease-symptom network}
\noindent
The ``Group
Selector'' function on the visualisation kit that HSG use, allows for the identification
of 19 such clusters in their disease-disease network, with each cluster
corresponding to a disease-class. 
This function also allows identification of
the number of diseases (i.e. nodes) in each disease-class (see left panel of
Figure~\ref{fig:compgraph}). The right
panel of Figure~\ref{fig:compgraph} displays the ratio of intra-class variance 
to the inter-class variance of each disease-class; the value of the area under
the Receiver Operating Characteristic curve (ROCAUC) for each cluster is
opverplotted, where the ROCAUC value for the $i$-th cluster can be
interpreted as the probability that a randomly chosen node is ranked as more
likely to be in the $i$-th class than in the $j$-th class, with $i\neq j;\:
i,j=1,\ldots,19$ \ctp{caspian}.

\begin{figure}[!hbt]
\centering
\includegraphics[width=13cm,height=6cm]{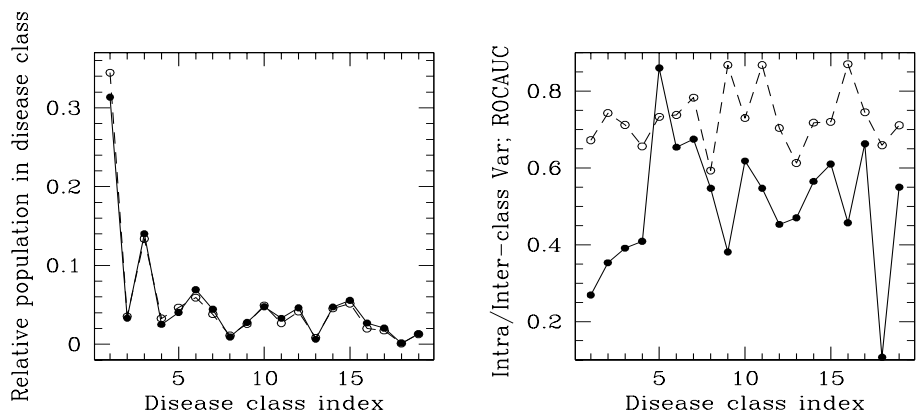}
\vspace{-.5cm}
\caption{{\it Left:} comparison of the
  relative number of nodes (diseases) that we recover in each of the 19
  disease classes that HSG classify their reported network to
  be classified into, with the relative class-membership reported by HSG. Our
  results are shown as filled circles joined by solid lines. In open circles
  threaded by broken lines, we overplot the relative
  number of diseases in each of the 19 classes, as reported by HSG. Similarity of the relative
  populations in the different disease classes, indicate that our learnt
  clustering distribution is similar to that obtained by HSG. {\it Right:}
our computed ratios of the averaged intra-class to inter-class variance for
each of the 19 classes, shown in filled circles; the ROC Area Under Curve
values reported by HSG for each class, is overplotted as open
circles joined by broken lines. The disease class
indices, from assigned values of 1 to 19, are the following respectively:
cellular proliferation diseases, integumentary diseases, diseases of the
nervous system, genetic diseases, diseases of metabolism, diseases by
infectious agents, diseases of mental health, physical disorders, diseases of
the reproductive system, of the immune system, of the respiratory system, of
the muscleoskeletal system, syndromes, gastrointestinal diseases,
cardiovascular diseases, urinary diseases, viral infections, thoracic
diseases, diseases of the endocrine system.}
\label{fig:compgraph}
\end{figure}

\section{Conclusion}
\noindent
In this work, we present a methodology that allows for the Bayesian learning
of the inter-column correlation of a rectangularly-shaped dataset, along with
uncertainties, and this in turn allows for the learning of the
with-uncertainties graphical model of such data, to then ultimately permit
computing the distance between a pair of such learnt graphical
models, of respective datasets. This novel, eventual computation of the
inter-graph distance--or rather of the distance between the posterior
probability of the graphs given the data--is important in the sense that it
informs on the correlation between datasets that are higher-dimensional
than being rectangularly-shaped, eg. correlation amongst slices of
rectangularly-shaped data, that together comprise a cuboidally-shaped dataset,
where each such rectangular slice of data is generated under distinct
experimental conditions. Then, the distance between the graphical models of a
pair of such slices of data, will inform us about the correlation between such
slices of data. Such information is easily calculable under the approach
discussed herein, even when the datasets are differently sized, and highly
multivariate. One example of such a situation could be a large network
observed on a sample of size $n_1$ before an intervention/treatment, and after
the implementation of such intervention, when a smaller sample (of size $n_2$;
$n_2\neq n_1$) is investigated. We illustrate the application of this method
on computing the distance between the uncertainty-accompanied, learnt
vino-chemical graphical models of Portuguese red and white wine
samples. Importantly, this example demonstrates that the two strands of this
work--namely learning graphical models with uncertainties, and computing
inter-graph distance--are indeed integrated.

This Bayesian approach allows for acknowledgement of errors of measurement
of any observable. The effect of ignoring such existent measurement errors, on
the learning of the between-columns correlation matrix, and ultimately on the
graphical model, is demonstrated using a simple, low-dimensional
simulated dataset (see Section~1.2 of the Supplementary Material). Even in such
a low-dimensional example, the difference made to the inferred graph of the
given data, by the inclusion of measurement errors, is clear.

Interestingly, we do not need to resort to the assumption of
decomposability in the MCMC-based inference that we use; to be precise, inference is performed with
Metropolis-within-Gibbs in which the correlation matrix is
first updated given the data, and the graph is then updated at the freshly
updated correlation, where we employ the closed-form likelihood for the
between-column correlation matrix, that we have achieved, (by marginalising
over all between-row correlation matrices). 

Our method is equally
capable of learning very large networks, as we have illustrated by undertaking
the learning of the human disease-symptom network (with $\geq$80,000 nodes).
When faced with the task of learning very large networks, i.e. a very
high-dimensional correlation matrix and a large number of edge parameters, we
can avoid undertaking the MCMC-based inference (that we adopt in general), as
long as the correlation structure is empirically known. This is often possible
when the problem of learning the correlation can be cast into a semantic
context--as was done in one of the applications that we considered, in
learning the very large human disease-symptom network that is marked by
disease-disease correlation in terms of the associated symptoms, ordered by
relevance. Other situations also admit such possibilities, for example, the
product-to-product, or service-to-service correlation in terms of associated
emotion, (or some other response parameter), can be semantically gleaned from
the corpus of customer reviews uploaded to a chosen internet facility, and the
same used to learn the network of products/services. Importantly, this method
of probabilistic learning of small to large networks, is useful for the
construction of networks that evolve with time, i.e. of dynamic networks.

\begin{supplement}
\sname{Supplement A}\label{suppA} 
\stitle{Supplementary Section for ``Learning of Correlation Structure $\&$ Random Graphs along
  with Uncertainties, to Compute Inter-Graph Distance''}
\sdescription
{All content of the supplementary material are referred to at
  relevant points in the text above.}
\end{supplement}

\renewcommand\baselinestretch{1.}

\newpage

\begin{center}
{\bf{\LARGE{Supplementary Section for ``Correlation between Multivariate Datasets, from Inter-Graph Distance computed using Graphical Models Learnt With Uncertainties''}}}
\end{center}


\noindent
Throughout, we refer to our main manuscript as WC.

\section{Empirical illustration: simulated data}
\label{sec:toy}
\noindent
The simulated data that we use in this section, is a 5-columned data set
$\bD_{orig}$ ($p$=5) with number of rows $n_{orig}=4000$, where
$\bD_{orig}$ is simulated to bear a chosen between-columns correlation matrix
$\bSigma_C^{(true)}$ that is given as:
\begin{equation}
\begin{pmatrix}
1 & 0.9914 & -0.8964 & 0.02526     & 0.0656\\
  & 1      & -0.8916 & 0.01981     & 0.6647\\
  &        & 1       & -0.009747   & -0.06140\\
  &        &         & 1           & 0.03622\\
  &        &         &             & 1
\end{pmatrix} \nonumber 
\end{equation}
which when inverted, allows for the computation of the empirical partial correlation matrix, following Equation~2.6 of WC (equation that gives the posterior of the between-columns correlation matrix given the data). This empirical partial correlation matrix is $\bR^{(true)}$:
\begin{equation}
\begin{pmatrix}
1 & 0.9574 & -0.2114 & 0.004786     & 0.005037\\
  & 1      & -0.04897 & 0.03900     & 0.01206\\
  &        & 1       & 0.02736     & -0.006288\\
  &        &         & 1           & 0.03527\\
  &        &         &             & 1
\end{pmatrix} \nonumber 
\end{equation}

We randomly sample $n$ (=300 typically) rows from this
simulated data set $\bD_{orig}$, to define our toy data set ${\bf D}_T$, that
we will implement in our method, to
\begin{enumerate}
\item[--] learn the between-columns correlation matrix
  $\bSigma_C^{(S)}=[S_{ij}]_{i=1;j=1}^{n,p}$ given the standardised version
  ${\bf D}_T^{(S)}$ of ${\bf D}_T$, and thereafter, learn the graphical model
  of data ${\bf D}_T^{(S)}$, as defined in Definition~2.1 of WC with $p$=5
  and partial correlation matrix $\bR=[R_{ij}]_{i=1;j=1}^{n,p}$, where
  elements of $\bR$ are computed using the learnt $\bSigma_C^{(S)}$ in Equation~2.6 of WC (posterior of between-columns correlation matrix given data). 
Here ${\bf D}_T^{(S)}$ comprises $n$ simulated values of the variables $Z_1,\ldots,Z_5$. 
\item[--] perform model checking using ${\bf D}_T^{(S)}$. To be precise, we
  predict the distribution of $Z_i$ when in the identified test data, $Z_{j}$ is restricted to take values
  in the chosen, narrow interval $[z_{j}^{(0)}-\delta_{j},
  z_{j}^{(0)}+\delta_{j}]$, for $j\neq i;\: i,j=1,\ldots,5$--and then compare
  the empirical distribution of $Z_i$ in the test data, with the posterior
  predictive distribution of $Z_i$, given the correlation matrix learnt using
  ${\bf D}_T^{(S)}$. Also, given ${\bf D}_T^{(S)}$ and $Z_{j}$, we perform
  MCMC-based sampling from the joint posterior of $\{Z_i\}_{i=1;i\neq
    j}^{i=p}$ and $\bSigma_C^{(S)}$. This is discussed in Section~1 of the
  Supplementary Section.
\item[--] learn the correlation matrix and graphical model of the data, where a chosen measurement error is placed on $Z_i$, $i=1,\ldots,p$; the unknown variance $v_{\epsilon_i}$ of this error density is also learnt. 
\end{enumerate}
Plots of $Z_i$ against $Z_1$ are included in Figure~\ref{fig:toydata_dat}; $i=2,3,4,5$.

\begin{figure}[!ht]
\centering
{\includegraphics[width=15cm,height=4cm]{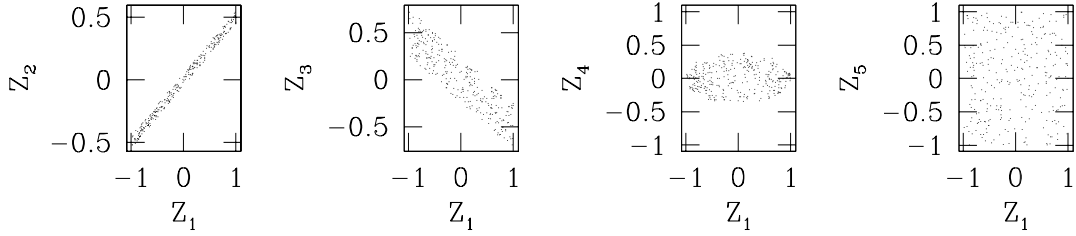}}
\caption{Plots of $Z_i$ against $Z_1$ in the standardised version of the toy
  data ${\bf D}_T^{(S)}$ simulated to bear the empirical column-correlation
  matrix $\bSigma_C^{(true)}$; here $i=2,3,4,5$. The toy data ${\bf D}_T^{(S)}$
  that we use in our work, comprises $n$ measurements of the variables
  $Z_1,...,Z_5$, with a typical $n$ of 300.}
\label{fig:toydata_dat}
\end{figure}

\subsection{Learning correlation matrix $\&$ graph given toy data ${\bf D}_T^{(S)}$}
\noindent
We learn the between-columns correlation matrix $\bSigma_C^{(S)}$ given the standardised toy data ${\bf D}_T^{(S)}$ by employing the algorithm discussed in Section~2 of WC. We use $n=300,\:p=5$, and with the aim of estimating the normalisation ${\hat{c}}_t$ of the posterior in the $t$-th iteration, we choose $K=20$ number of sampled data sets with $n^{/}$ rows and $p$ columns, generated in each iteration, to bear the column-correlation matrix proposed in that iteration. Indeed, we set $n^{/}=n$. Here $t=0,\ldots,N$.

\begin{figure}[!t]
\centering
\includegraphics[width=12cm,height=15cm]{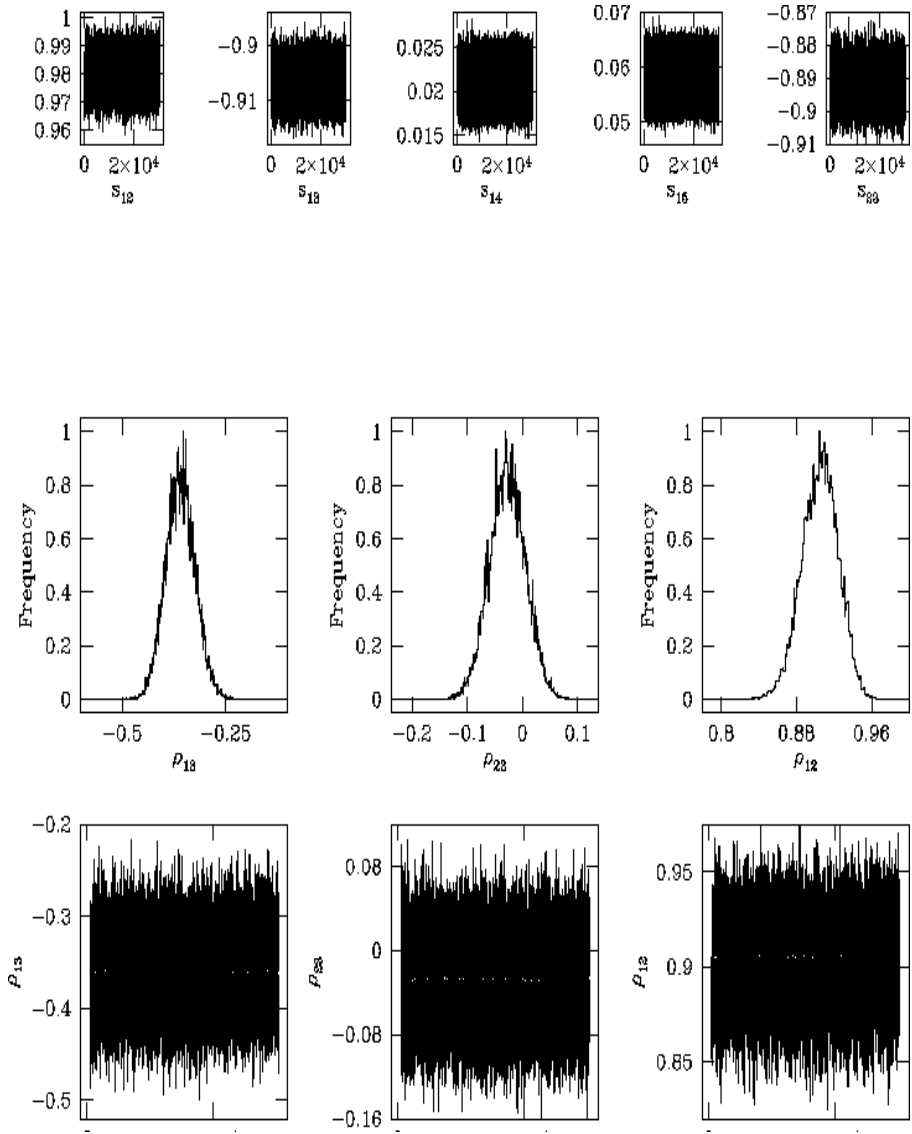}
\caption{Figure showing traces and marginal posterior probability densities
  (as histograms) of elements of the correlation matrix $\bSigma_C^{(S)}$, and
  partial correlation matrix $\bR$,  learnt given the toy data ${\bf
    D}_T^{(S)}$, in our method in which the data is modelled using a
  matrix-variate Gaussian Process, and the likelihood obtained by
  marginalising over the between-row correlation matrix. The top panel
  displays traces of the five correlation parameters
  $s_{12},s_{13},s_{14},s_{15},s_{23}$ given this toy data. The lower-most
  panel displays traces of the partial correlation parameters $\rho_{12}$,
  $\rho_{13}$, $\rho_{23}$, computed using correlation matrix
  $\bSigma_C^{(S)}$ learnt given ${\bf D}_T^{(S)}$, in
  Equation~2.6 of WC. The middle panel presents the marginals of these partial correlation parameters as histograms. }
\label{fig:toydata}
\end{figure}

In the $t$-th iteration of our MCMC chain, the first block update in our Metropolis-within-Gibbs inference scheme, leads to the updating of the column correlation matrix to $\bSigma_t$ given the data ${\bf D}_T^{(S)}$, using which we compute the value of the partial correlation matrix $\bR_t=[\rho_{ij}^{(t)}]$ in this iteration. Then the second block update leads to the updating of the values of the binary graph edge parameters to $g_{ij}^{(t)}$ and variance parameters to $\sigma_{ij}^{(t)}$, given $\bR_t$. Traces of the marginal posterior probability of five of the $S_{ij}$ parameters given data ${\bf D}_T^{(S)}$ are shown in the top left panel Figure~\ref{fig:toydata}, while the joint posterior of all $G_{ij}$ and $\sigma_{ij}$ parameters given the learnt partial correlation matrix, is shown in the top left panel Figure~\ref{fig:toydata_var}. Histograms representing approximations of marginals of individual $R_{ij}$ and $\sigma_{ij}$ parameters, given the data and the learnt partial correlation respectively, occupy other panels of Figure~\ref{fig:toydata} and Figure~\ref{fig:toydata_var} respectively. 
Here $i < j; i,j=1,\ldots,p$. 

\begin{figure}[!ht]
\centering
\includegraphics[width=8cm,height=8cm]{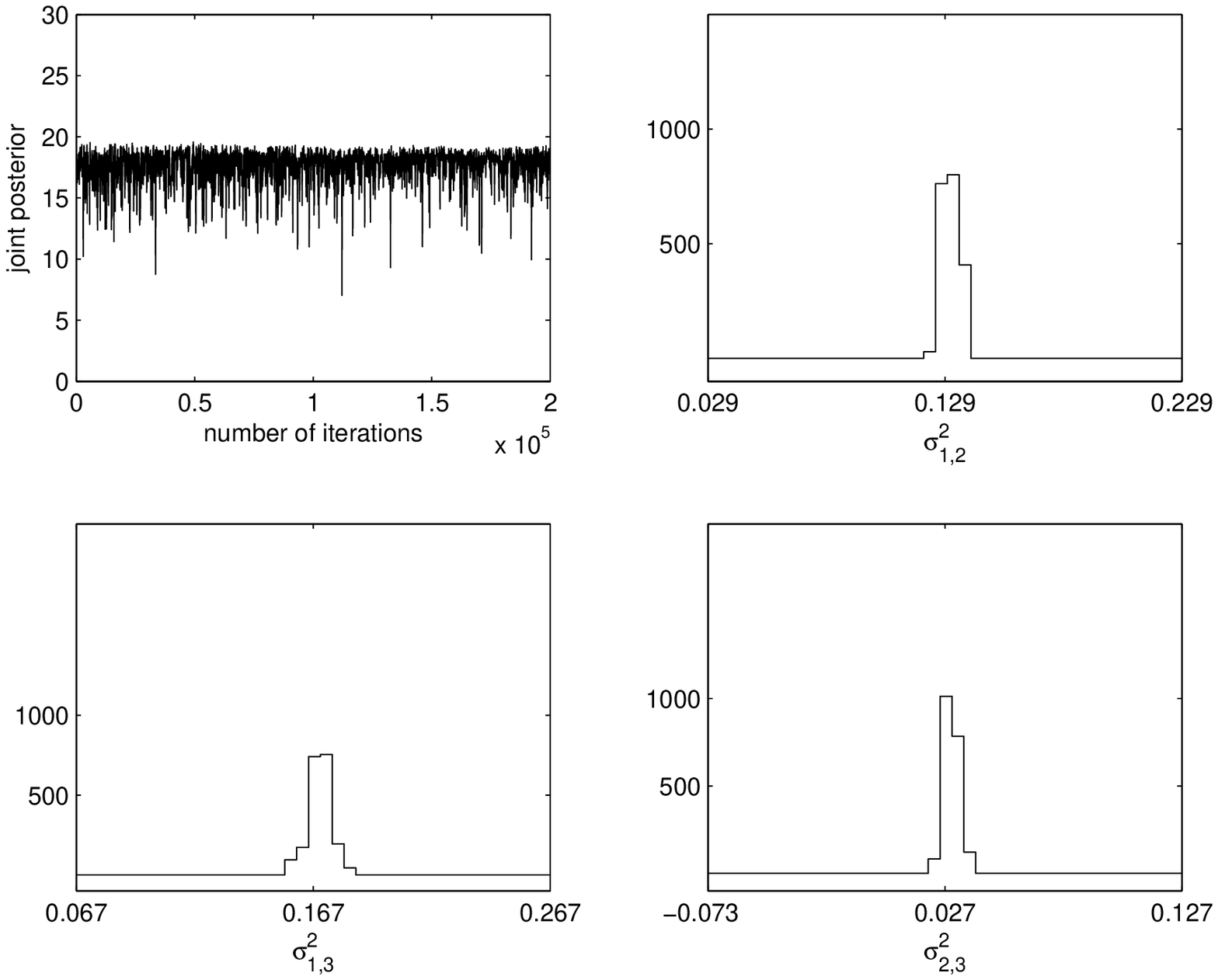}
\caption{{\it Top left:} trace of joint posterior probability density of the graph edge parameters $g_{ij}$ and variance parameters $\sigma^2_{ij}$, given the partial correlation matrix learnt in the first block update of our Metropolis-within-Gibbs inference scheme, given the 5-columned toy data set ${\bf D}_T^{(S)}$. {\it Other panels:} histogram approximations to the marginal posterior probability density of three of the variance parameters.}
\label{fig:toydata_var}
\end{figure}

The graphical model of the data ${\bf D}_T^{(S)}$ is presented in Figure~\ref{fig:toydata_graph}. The fraction $n_{ij}$ of post-burnin samples of $g_{ij}$ with a value of 1, i.e. an approximation to the probability of existence of the edge joining nodes $i$ and $j$, is marked next to each edge of the graph, as long as $n_{ij} \geq 0.05$, i.e. the edge probability parameter $\phi_{ij}(R_{ij})$ is non-zero. 

\begin{figure}[!ht]
\centering
\includegraphics[width=8cm,height=8cm]{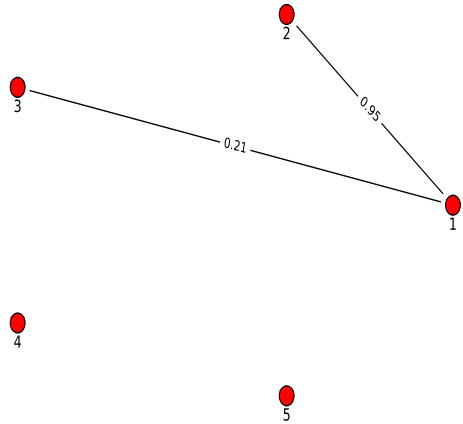}
\vspace{-.1cm}
\caption{Figure showing graphical model of toy data ${\bf D}_T^{(S)}$--learnt in our Metropolis-within-Gibbs inference scheme in which we learn the correlation matrix $\bSigma_C^{(S)}$ of the data, simultaneously with the graph. The observables $Z_1,...Z_5$, measurements of which comprise the data, are marked by filled red circles, as the 5 nodes in this graph. The probability of the edge parameter $g_{ij}$ to exist (i.e. for $g_{ij}$ to be 1)--$i\neq j$, $i,j=1,\ldots,5$--is approximated by the fraction $n_{ij}$ of post-burnin iterations in which the current value of $g_{ij}$ is 1. This value of $n_{ij}$ is marked against the edge joining the $i$-th and $j$-th nodes, as long as $n_{ij} > 0.05$. }
\label{fig:toydata_graph}
\end{figure}

We note that the column correlation matrix $\bSigma_C^{(S)}$ of the Gaussian Process that models the data, is such that the partial correlation $\rho_{12}$ between $Z_1$ and $Z_2$ is learnt to be in the 95$\%$ HPD credible region of $\in[0.86, 0.95]$ approximately, which is close to the empirical value of 0.96. Again, the empirical value of $\rho_{13}$ is about -0.2, and the learnt value is $\in[-0.44, -0.27]$ approximately; empirical value of $\rho_{23}$ is about 0.04, and the learnt value is $\in[-0.11, 0.05]$ approximately. The other partial correlation parameters have smaller values in the chosen correlation structure that the data is simulated to bear--each of which is close to the corresponding learnt value. This offers confidence in our method of learning the correlation matrix $\bSigma_C^{(S)}$ of the standardised toy data ${\bf D}_T^{(S)}$.

\subsection{Incorporating measurement uncertainties in the learnt graphical model}
\label{sec:noise}
\noindent
If measurement errors affect the values of the $i$-th component $Z_i$
of the $p$-dimensional vector-valued observable $\bZ$, where measurements of $Z_i$ comprise the $i$-th column of data ${\bf D}_S$,
($i=1,\ldots,p$), the variance of the probability distribution of
such errors--if unknown--can be learnt given the data. So let the
error in $Z_i$ be $\epsilon_i$ that we assume is Normally distributed with
variance $v_{\epsilon_i}$,
i.e. $\epsilon_i\sim{\cal N}(0,v_{\epsilon_i})$. Then if the unknown
error variance $v_{\epsilon_i}$ is proposed in the $t$-th iteration of
our MCMC chain to be $v_{\epsilon_i}^{(t\star)}$, the correlation
$s_{ij}^{(t\star)}$ has to be adjusted by the factor
$1/\sqrt{1+v_{\epsilon_i}^{(t\star)}}$, $\forall j\neq i$.

So, in the presence of measurement error in $X_i$, the absolute value of the correlation $s_{ij}$ between $Z_i$ and $Z_j$ decreases (by a factor of $\sqrt{1 + v_{\epsilon_i}}$ in the model in which variances add linearly). 

\begin{figure}[!hbt]
\centering
\includegraphics[width=13cm,height=10cm]{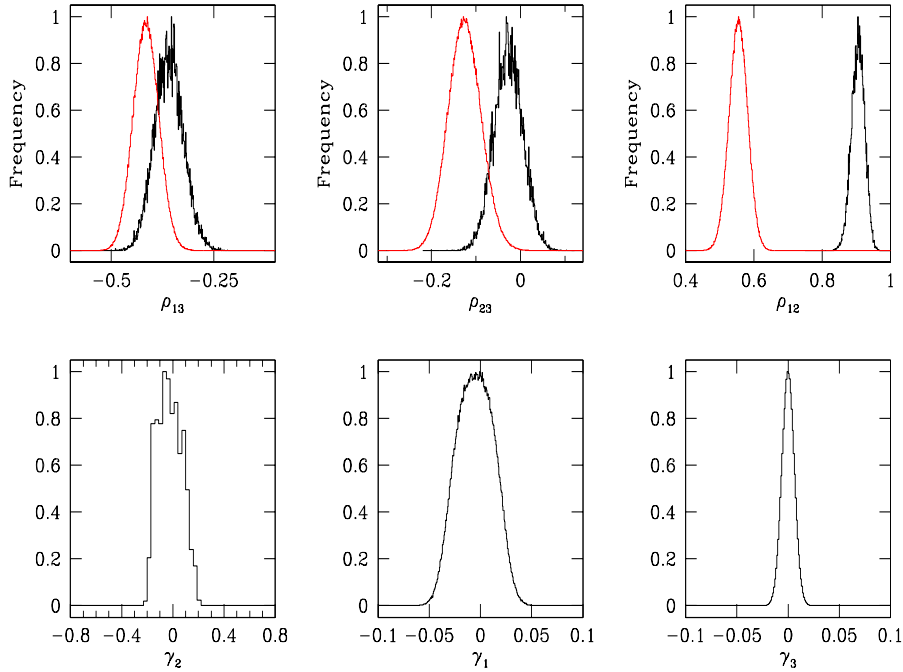}
\caption{{\it Top panels:} comparison of histogram representation (in black)
  of the marginal posterior density of some partial correlation parameters
  ($\rho_{ij}$) learnt given toy data ${\bf D}_T^{(S)}$, with the marginals
  (in grey, or red in the electronic version), of the same parameter, learnt 
given the data ${\bf D}_T^{(err)}$, which differs from ${\bf D}_T^{(S)}$, in only that Gaussian errors of variance 0.01 are imposed on the variable $Z_2$. i.e. the 2nd component of the 5-dimensional observable vector $(Z_1,Z_2,Z_3,Z_4,Z_5)^T$, measurements of which comprise the data. Here $i,j=1,...,5; i\neq j$. From left to right, are presented the results for $\rho_{12}, \rho_{13}$ and $\rho_{23}$. 
{\it Lower panels:} histogram representations of the standard deviation
$\gamma_i$ of the error density in the measurement of $Z_i$, learnt using data ${\bf D}_T^{(err)}$, for $i=2,1,3$ from the left to the right panels, where in this data, $Z_2$ is the only one of the 5 variables that has an error (of standard deviation 0.1) imposed on it. 
} 
\label{fig:toydata_err}
\end{figure}

On the other hand, the partial correlation $\rho_{ij}$ may increase or
decrease \ctp{liu}. That such is a possibility, is corroborated in
the correlation and partial correlation structures of an example data set that
comprises measurements of a 3-dimensional observable vector
$(Z_1,Z_2,Z_3)^T$. Then, 
$\rho_{ij}=\displaystyle{\frac{s_{ij}-s_{ik}s_{jk}} {\sqrt{(1-s_{ik}^2)(1-s_{jk}^2)}}}$,
$i\neq j, i\neq k, k\neq j; i,j,k=1,2,3$. It follows that if
$\vert{s_{ij}}\vert$ and $\vert s_{ik}\vert$ decrease, $\rho_{ij}$ can
either increase or decrease. But $\rho_{ij}$ is the probability for
the edge between the $i$-th and $j$-th nodes of the graph of this
data, to exist, i.e. $\rho_{ij}=\Pr(g_{ij}=1)$.  Then it is possible
that while in the absence of measurement errors, $g_{ij}=1$ during a
fraction $n_{ij} <0.05$ of the number of post-burnin iterations, in the
presence of measurement error in $X_i$, $\rho_{ij}$ increases sufficiently to
ensure that the fraction of iterations during which this edge exists
is in excess of 0.05. If this happens, the edge between the
$i$-th and $j$-th nodes will be included in the graphical model of the
data when measurement error in $X_i$ is acknowledged, but not when
such error is not. In other words, ignoring measurement uncertainties
can lead to a potential misrepresentation of the graphical model of
the data at hand.

\begin{figure}[!ht]
\centering
\includegraphics[width=8cm,height=8cm]{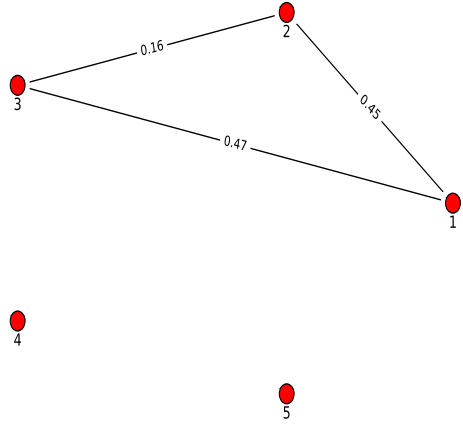}
\vspace{-.7cm}
\caption{Figure showing graphical model of data ${\bf D}_T^{(err)}$ that
  differs from the toy data ${\bf D}_T^{(S)}$ only in that Gaussian errors
  (with variance 0.01) are added to the 2nd column of ${\bf D}_T^{(S)}$, to
  realise ${\bf D}_T^{(err)}$. The inclusion of measurement noise in this
  column of the toy data is noted in the learnt graphical model of the
  resulting error-bearing data ${\bf D}_T^{(err)}$, which manifests the edge
  between variables $Z_2$ and $Z_3$, while this edge is absent in the graphical
  model of the error-free data ${\bf D}_T^{(S)}$; see Figure~\ref{fig:toydata_graph}. }
\label{fig:toydata_graph_err}
\end{figure}

In our work, it is possible to produce graphs while ignoring, as well as acknowledging the measurement uncertainty in one or more components of the $p$-dimensional observable vector, $n$ measurements of which results in the rectagularly-shaped data at hand. In fact, it is also possible to learn the variance of the error density of the components of this obsrvable. We demonstrate this in the experiment discussed here.
  
In this implementation, we add measurement error to the 2nd component $X_2$ of
the 5-dimensional observable vector, $n$ standardised measurements of which
comprise data ${\bf D}_T^{(S)}$. We choose to impose Gaussian measurement errors on $Z_2$,
s.t. this Gaussian error density is $\epsilon_2\sim{\cal N}(0,0.01)$.  We then
define a data set that is the same as ${\bf D}_T^{(S)}$, except that the 2-nd
column of this data is now sampled from a Gaussian with zero mean and variance
given by 1+0.01, i.e. sampled from the convolution of a standard Normal, with
the density ${\cal N}(0,0.01)$. The resulting data set is referred to as ${\bf
  D}_T^{(err)}$. Thus, the true value of the variance $v_{\epsilon_2}$ of the
2nd column of the data ${\bf D}_T^{(err)}$ is 0.01. We will treat this
variance as an unknown and in fact, learn this value using ${\bf D}_T^{(err)}$.

We learn the column-correlation matrix of this data using the method
delineated in Section~2 of WC, using an MCMC chain that we
run with this data ${\bf D}_T^{(err)}$. The only exception to the
method of learning the $s_{ij}$ parameters is that the correlation
between the $Z_i$ and $Z_j$ is given by
$\displaystyle{\frac{s_{ij}}{\sqrt{(1 + v_{\epsilon_i})(1 +
      v_{\epsilon_j})}}}$ in the model in which the variances are
assumed to add linearly; $i\neq j; i,j=1,\ldots,p$. Thus, in addition
to the $p(p-1)/2$ number of $s_{ij}$ parameters, we now also learn the
$p$ number of $v_{\epsilon_i}$ parameters, where the latter is the
variance of the error distribution of $Z_i$. We actually learn the
standard deviation of the error density on $Z_i$, namely $\gamma_i$,
i.e. $v_{\epsilon_i} = \gamma_i^2$. In the $t$-th iteration, we propose
$\gamma_i$ from a Gaussian proposal density that has the mean given by
the current value of the parameter in this iteration, and an
experimentally chosen variance. Here $t=0,\ldots,N$. This is
undertaken $\forall i=1,\ldots,p$. The $S_{ij}$ parameters are
always proposed from Truncated Normal proposal densities that are
left and right truncated at -1 and 1 respectively and have mean given
by the current parameter value, while the variance is fixed. Then the
correlation parameters that define the correlation matrix in the
$t$-th iteration, are $s_{ij}^{(t\star)}/{\sqrt{(1 +
    (\gamma_{\epsilon_i}^{(t\star)})^2)(1 + (\gamma_{\epsilon_j}^{(t\star)})^2)}}$,
$i\neq j; i,j=1,\ldots,p$. We use Gaussian priors on the $S_{ij}$
parameters, where such a Gaussian is centred on the empirical
correlation between $Z_i$ and $Z_j$ in the data, while uniform priors
are used on all other parameters. Using the proposed and current
correlation matrices in our Metropolis-Hastings inferential scheme, we
compute the marginals of the individual $S_{ij}$ parameters as well as
the $\gamma_i$ parameters ($\gamma_i^2 = v_{\epsilon_i}$).

Histogram representations of the marginals (normalised to 1 at the mode), of
some of these parameters are displayed in Figure~\ref{fig:toydata_err}. The
95$\%$ HPD credible region on $\gamma_2$ that we learn given this data is
[-0.2,0.2] approximately. The learnt standard deviations of the error
densities of variables other than $Z_2$, are 0 approximately. We also note
from this figure that the changes in the partial correlations introduced by
the introduction of the measurement error in one variable, can be both an
increase and decrease--this is discussed above. The effect on introducing this
measurement error on $Z_2$, on the graphical model of the data ${\bf
  D}_T^{(err)}$, is presented in Figure~\ref{fig:toydata_graph_err}. In this
graphical model, the edge $G_{23}$ between the 2-nd and 3-rd nodes takes the
value 1, with probability of about 0.16, while $n_{23}$ was less than 0.05 in
the graphical model of data ${\bf D}_T$--which differs from ${\bf
  D}_T^{(err)}$ only in that the 2nd column is imposed with a Gaussian error
of variance 0.01. Thus, the effect of introducing this error to measurements of
the variable $Z_2$ propagates into the (partial) correlation structure of the
data, to then affect the graphical model. Comparing this learnt graph to the graph
of the toy data ${\bf D}_T^{(S)}$, we recognise that measurement errors can
distort the graphical model of a data.

\section{Model checking}
\noindent
In the Section~2 of WC, we discussed the learning of
$\bSigma_C^{(S)}$ using the $n$ rows of the standardised toy data
${\bf D}_T^{(S)}$, which is a 300-row subset from the 5-columned
simulated dataset $\bD_{orig}$, discussed in the previous section,
where $\bD_{orig}$ is generated to abide by a chosen correlation
matrix $\bSigma_C^{(true)}$ that is defined above in
Section~\ref{sec:toy}. Then ${\bf D}_T^{(S)}$ comprises 300 different
measurements of the 5-columned vector $\bZ:=(Z_1, Z_2, Z_3, Z_4,
Z_5)^T$, where $Z_i$ is a standardised variable $i=1,\ldots,5$.
Having learnt the parameters of the Gaussian Process in
Section~\ref{sec:toy}--of which the standardised observable
$\bZ\in{\mathbb R}^p$ is a realisation--here we want to predict values
of $Z_i$ for values of $Z_j$ as given in a new or test data, ($j\neq
i;\: i,j=1,\ldots,p$); for our purposes, $p$=5. This test data ${\bf
  D}_{test}$ is built to be independent of the training data ${\bf
  D}_T^{(S)}$, as $q$ rows of the standardised version of the bigger
data set ${\bf D}_{orig}$--of which ${\bf D}_T^{(S)}$ is also a
subset--although the $q$ rows of ${\bf D}_{orig}$ that comprise ${\bf
  D}_{test}$, are chosen as distinct from the $n$ rows of the training
data ${\bf D}_{T}^{(S)}$. Our standardised test data ${\bf D}_{test}$
has $p=5$ columns and $q$ rows; in fact, we set $q=n$. We will predict
$Z_2, Z_3, Z_4$ at each of the known $q$ (=$n$) values of $Z_1$ in the
test data ${\bf D}_{test}$, given the GP parameters (i.e. the
between-columns covariance matrix $\bSigma_C^{(S)}$) that we learn
using the training data. No prediction of $Z_5$ is undertaken. In
fact, we will sample from the posterior predictive density of
$Z_2,Z_3,Z_4$, given the correlation matrix learnt using training data
${\bf D}_T^{(S)}$, and values of $Z_1$ in the test data ${\bf
  D}_{test}$. We compare the predicted values of $Z_2,Z_3,Z_4$ against
their empirical values in the test data. Such a comparison constitutes
the checking of our models s well as the results (of the learning of
$\bSigma_C^{(S)}$ given the training data ${\bf D}_{T}^{(S)}$). We
clarify this prediction now.

As we learn the marginal posterior probability density of each correlation
parameter $S_{ij}$ given ${\bf D}_{T}^{(S)}$, we need to choose a summary of
this marginal distribution, at which the prediction of the $z_{ik}$ is
undertaken, $i=2,3,4$, $k=1,\ldots,n$. We choose the mode of the marginal as
this summary. Denoting the value of $Z_i$ in the $k$-th row of the test data
as $z_{ik}$, ($k=1,\ldots,q=n$), we undertake the learning of
$\{z_{2k},z_{3k},z_{4k}\}_{k=1}^n$ in the test data ${\bf D}_{test}$, given
values of $\{z_{1k}\}_{k=1}^n$ in ${\bf D}_{test}$ and the modal values of
$S_{ij}$ learnt using the training data ${\bf D}_{T}^{(S)}$.  In our Bayesian,
MCMC-based inferential approach, this learning is equivalent to sampling from
the posterior predictive of the unknowns, i.e. performing MCMC-based posterior
sampling from
$$\pi(z_{21},z_{31},z_{41},\ldots,z_{2n},z_{3n},z_{4n}\vert z_{11},\ldots,z_{1n}, s^{(M)}_{12}, \ldots, s^{(M)}_{1p},s^{(M)}_{23}, \ldots, s^{(M)}_{2p},\ldots, s^{(M)}_{p-1\:p}),$$
where $s^{(M)}_{ij}$ represents the modal value of the correlation parameter $S_{ij}$ that we learn given the training data ${\bf D}_{T}^{(S)}$. We define the learnt ``modal'' correlation matrix to be $\bSigma_C^{(M)}=[s^{(M)}_{ij}]$.

In the $t$-iteration, we propose a value $z_{ik}^{(t\star)}$ from a Gaussian
proposal density with mean given by the current value $z_{ik}^{(t-1)}$ of this
variable, and fixed variance $\nu_{ik}$, i.e. the proposed value is
$z_{ik}^{(t\star)}\sim {\cal N}(z_{ik}^{(t-1)}, \nu_{ik})$; we do this for
$i=2,3,4$ and $\forall k=1,\ldots,n$, at each $t=0,\ldots,N$. Then the
proposed data in the $t$-th iteration is ${\bf
  D}^{(t\star)}=\left(\bz_{1},\bz_{2}^{(t\star)}, \bz_{3}^{(t\star)},
  \bz_{4}^{(t\star)}, \bz_{5}\right)$, where
$\bz_{i}=(z_{i1},\ldots,z_{in})^T$, $i=1,\ldots,5$. The posterior of the unknowns is then
given as in Equation~2.8, with the data given by ${\bf
  D}^{(t\star)}$ and the modal correlation matrix given by $\bSigma_C^{(M)}$
learnt using the training data set ${\bf D}_{T}^{(S)}$. The normalisation of
the posterior is computed in the $t$-th iteration in the way described in
Section~2.3 of WC, at the $\bSigma_C^{(M)}$. We use uniform priors on all unknowns. So in
each iteration, we (use Random-Walk Metropolis to) sample from the 
posterior of the unknown variables, given $\bSigma_C^{(M)}$ and
the data on the $q=n$ number of $Z_{1}$ values in the test data
${\bf D}_{test}$. We implement such posterior sampling to compute marginal
predictive of each of the unknowns. We compare this marginal predictive of of
$Z_{2}, Z_{3},Z_{4}$, to the empirical distribution of $Z_{2},Z_3,Z_{4}$ in
the test data ${\bf D}_{test}$. We also compare the plots of the predicted
$Z_{i}$ and the known $Z_{1}$ values, to the corresponding plot of empirical
value of $Z_i$ and $Z_1$; $i=2,3,4$. The results of this comparison for $Z_2,
Z_3$ and $Z_4$ are included in Figure~\ref{fig:toydata_modchk1}.

Figure~\ref{fig:toydata_modchk1} shows that the plots of the predicted values of $Z_i$,
$i=2,3,4$, against $Z_1$ (in red filled circles in the electronic version, and
grey circles in the monochrome version), compare favourably--visually
speaking--to the plots of the empirical $Z_i$ (in the test data), against
$Z_1$. To be precise, the red (or grey) circles comprise predicted (or learnt)
pair ($z_{1k}, z_{ik}^{(mode)}$) for $k=1,\ldots,q=n$, where $z_{ik}^{(mode)}$
is the modal value of the marginal posterior density of $Z_{ik}$ given known
values of $Z_1$ in the test data, and the (modal) correlation matrix
$\bSigma_C^{(M)}$ (itself learnt given the training data). The black circles
represent the empirical values ($z_{1k}, z_{ik}$) for $k=1,\ldots,n$, i.e. the pair in the $k$-th row of the test data. We also plot the marginal of the learnt values of $Z_{i}$ given the data, superimposed on the frequency distribution of the empirical value of $Z_{i}$ in the test data--we do this for each $i=2,3,4$. Again, the overlap between the results is encouraging. Thus, the predictions offer confidence in our model, as well as the results of our learning of the correlation structure of the data. 


\begin{figure}[!h]
\centering
\includegraphics[width=13cm,height=8cm]{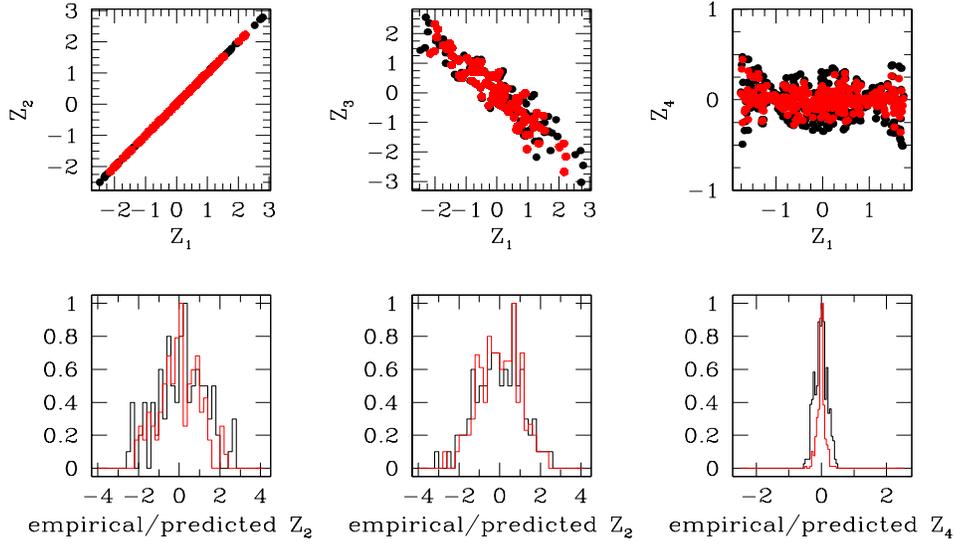}
\caption{{\it Top panels:} figures comparing plots of empirical and predicted
  values of $Z_i$ against values of $Z_1$, for $i=2,3,4$ moving from left to
  the right panel. Grey (red in the electronic version) circles depict pairs of
  $(z_{1k}, z_{ik})$ in the test data ${\bf D}_{test}$, while black circles
  depict $Z_{i}$ values learnt given the first column of the test data and the
  modal correlation matrix $\bSigma_C^{(M)}$ that is itself learnt using the
  training data set ${\bf D}_{T}^{(S)}$. {\it Lower panels:} marginal of $Z_{i}$ given 1st column of test data and $\bSigma_C^{(M)}$, plotted as a histogram in grey (or red in the electronic version), over its empirical distribution in black, i.e. the histogram of the $i$-th column of the test data. Here, $i=2,3,4$ as we move from left to right.
}  
\label{fig:toydata_modchk1}
\end{figure}

However, conditioning the posterior predictive of $Z_i$ on a summary--modal in
our earlier implementation--correlation matrix learnt given training data
${\bf D}_T^{(S)}$ is restrictive in that this approach ignores the learnt
distribution of the correlation matrices. After all, our learning of the
correlation matrix given ${\bf D}_T^{(S)}$ is MCMC-based, generating a value
of $\bSigma_C^{(S)}$ in each iteration. In light of this, the marginal
posterior of $Z_i$ obtained by marginalisation over the joint posterior
probability density of all unknown components of $\bZ$ and $\bSigma_C^{(S)}$
is a possibility. Thus, we learn
$\bSigma_C^{(S)}$ simultaneously with $Z_2,Z_3,Z_4$, i.e. the 2nd, 3rd and 4th columns of the test data, given the training data and the 1st column of the test data. We will then perform MCMC-based posterior sampling from the joint posterior probability density:
\begin{equation}
\pi\left(s_{12},\ldots,s_{1p},s_{23},\ldots,s_{2p},\ldots,s_{p-1\: p}, z_{21},\ldots,z_{2n},z_{31},\ldots,z_{3n},z_{41},\ldots,z_{4n}\vert z_{11}, z_{1n}, {\bf D}_T^{(S)}\right).
\label{eqn:jtmod}
\end{equation}
In order to implement this, we propose $z_{21}^{(t\star)},\ldots,z_{2n}^{(t\star)},z_{31}^{(t\star)},\ldots,z_{3n}^{(t\star)},z_{41}^{(t\star)},\ldots,z_{4n}^{(t\star)}$ in each of the $t$ iterations, $t=0,\ldots,N$. Each of these parameters is proposed from a Gaussian proposal density (with mean given by the current value and an experimentally chosen variance). At the same time, we propose the $s_{ij}$ parameters, $i\neq j$, $i,j=1,\ldots,p$ from a Truncated Normal proposal density, truncated at -1 and 1, with mean given by the current value of the parameter, and chosen variance. 

\begin{figure}[!ht]
\centering
\includegraphics[width=13cm,height=8cm]{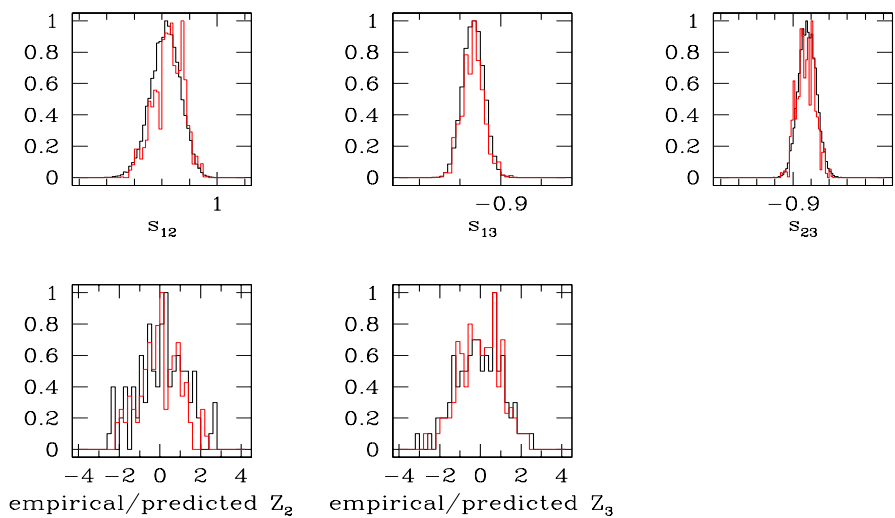}
\caption{{\it Top panels:} grey (red in the electronic version) coloured histograms represent the marginal posterior density of $S_{ij}$ learnt, (along with the $Z_{i}$ parameters; $i=2,3,4$), given the training data ${\bf D}_T^{(S)}$, and the known 1st column of the test data ${\bf D}_{test}$. This is compared to the marginal of $S_{ij}$ learnt (when the column-correlation matrix is learnt alone), given training data--presented as the histograms in black. Panels from left to right correspond to the results for $S_{12}$, $S_{13}$ and $S_{23}$ respectively. The lower panels present the comparison between the empirical distribution of the $i$-th column of the test data ${\bf D}_{test}$--in black--and the joint posterior of $Z_{i}$, (learnt along with the $S_{ij}$ parameters), given ${\bf D}_T^{(S)}$, and the 1st column of ${\bf D}_{test}$, (in grey, or red in the electronic version). Here $i=2$, in the bottom left panel and $i=3$ in the right.} 
\label{fig:toydata_modchk2}
\end{figure}

For this implementation, at the $t$-th iteration, we need to define the
augmented data ${\bf D}_A^{(t\star)}$, which is the training data ${\bf
  D}_T^{(S)}$, augmented by the data set ${\bf D}^{(t\star)}$ proposed in the
$t$-th iteration, (defined above), where the 1st and 5th columns of ${\bf
  D}^{(t\star)}$ are the known 1st and 5th columns of the test data ${\bf
  D}_{test}$, and the $i$-th column is the proposed vector
$(z_{i1}^{t\star},\ldots,z_{in}^{t\star})^T$, $i=2,3,4$. Thus, as the proposed
${\bf D}^{(t\star)}$ varies from one iteration to the next, the augmented data
${\bf D}_A^{(t\star)}$ also varies. This augmented data then has $p$ columns
nd $n+q$ rows, i.e. $2n$ rows, given our choice of $q=n$. In the $t$-th
iteration, the posterior probability density of the unknowns given this augmented data
${\bf D}_A^{(t\star)}$ is computed, using the posterior defined in
Equation~2.8 of WC in which the generic data ${\bf D}_S$ is now replaced by ${\bf D}_A^{(t\star)}$. While we impose uniform priors on the $z_{ik}$ parameters, we place Gaussian priors on $s_{ij}$, with such a prior centred at the empirical value of the correlation between the $i$-th and $j$-th columns of the data, ($i,j=1,\ldots,p$); the variance of these Gaussian priors are experimentally chosen. 

Some results of sampling from the joint defined in Equation~\ref{eqn:jtmod} are shown in Figure~\ref{fig:toydata_modchk2}. These include comparison of the histogram representations of the marginals of 3 correlation parameters $S_{12}, S_{13}, S_{23}$, learnt in this implementation given the augmented data, with the marginal of the same correlation parameter learnt given training data ${\bf D}_T^{(S)}$. The figure also includes a comparison of the empirical and predicted marginals of $Z_{2}$ and $Z_{3}$.

\section{Some results given the white wine data set}
\noindent
Figure~\ref{fig:white_var} presents trace of the joint posterior of
the $G_{ij}$ and $\sigma_{ij}^2$ parameters, updated in the 2nd block
of each iteration of our MCMC chain run with the white wine data, at
the updated (partial) correlation matrix. The other panels of this
figure depict the histogram representation of the marginals of some of
the $\sigma_{ij}^2$ parameters learnt given the white wine data.

\begin{figure}[!ht]
\centering
\includegraphics[width=13cm,height=10cm]{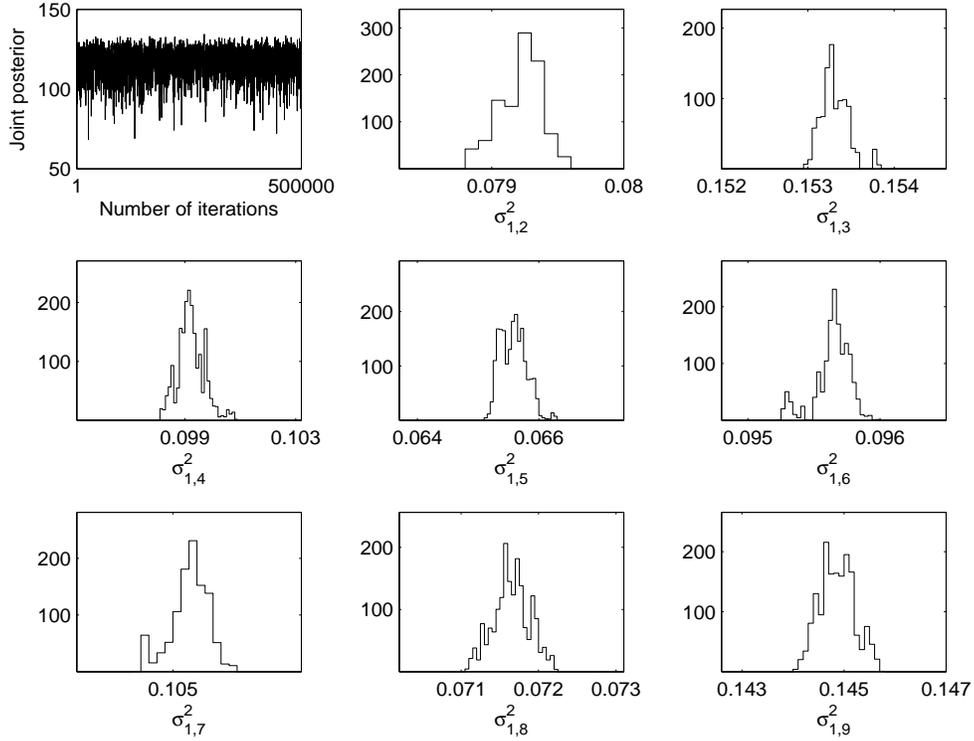}
\caption{{\it Top left panel:} trace of the joint posterior
  probability of the graph edge parameters $G_{ij}$ and the variance parameters $\sigma_{ij}^2$ that are the variances used in the likelihood function defined in Equation~2.11 of WC; these parameters are updated within the 2nd block update of our Metropolis-within-Gibbs inferential scheme, at the correlation matrix that is updated given the data ${\bf D}_S^{(white)}$ of Portuguese white wine samples. Here $i\neq j;\:i,j=1,\ldots,12$. {\it All other panels:} histogram representations of marginal posterior probability densities of some of the variance parameters learnt given the correlation matrix that is itself learnt, given data ${\bf D}_S^{(white)}$.
} 
\label{fig:white_var}
\end{figure}

\section{Comparing against previous work done with white wine data}
\label{sec:white}
\noindent   
The graphical model of the white wine data presented in
Fig~2 of WC is strongly corroborated by the simple empirical
correlations between pairs of different vino-chemical properties--this
correlation structure is apparent in the ``scatterplot of the
predictors'' included as part of the results of the ``Exploratory Data
Analysis'' reported in \\\url{https://onlinecourses.science.psu.edu/stat857/node/224} on the white wine data. They use the
full white wine data set ${\bf D}^{(white)}_{orig}$, to construct a
matrix of scatterplots of $X_i$ against $X_j$, where $i\neq
j;\;i,j=1,\ldots,11$. It is to be noted that in the data analysis reported in \url{https://onlinecourses.science.psu.edu/stat857/node/224}, the matrix of scatterplots of pairs of variables $i$ and $j$ was included, where this set of variables excluded the last column of the white wine data--the column that informs us of the assessed ``quality'' of the wine. 

When we compare our learnt graphical model with the results of this
reported ``Exploratory Data Analysis'', we remind ourselves that
partial correlation (that drives the probability of the edge between
the $i$-th and $j$-th nodes), is often smaller than the correlation
between the $i$-th and $j$-th variables, computed before the effect of
a third variable has been removed \ctp{sheskin}. If this is the
case, then an edge between nodes $i$ and $j$ in the learnt graphical
model, is indicative of a high correlation between the $i$-th and
$j$-th variables in the data. However, in the presence of a suppressor
variable (that may share a high correlation with the $i$-th variable,
but low correlation with the $j$-th), the absolute value of the
partial correlation parameter can be enhanced to exceed that of the
correlation parameter. In such a situation, the edge between the nodes
$i$ and $j$ in our learnt graphical model may show up (within our
defined 95$\%$ HPD credible region on edge probabilities, i.e. at
probability higher than 0.05), though the empirical correlation
between these variables is computed as low \ctp{sheskin}. So, to
summarise, if the empirical correlation between two variables reported
for a data set is high, our learnt graphical model should include an
edge between the two nodes. But the presence of an edge between pair
of nodes is not necessarily an indication of high empirical
correlation between a pair of variables--as in cases where suppressor
variables are involved. Guessing the effect of such suppressor
variables via an examination of the scatterplots is difficult in this
multivariate situation. Lastly, it is appreciated that empirical trends are
only indicators as to the matrix-Normal density-based model of the learnt correlation structure (and the graphical model learnt thereby) given the data at hand. 

\section{Results of learning given the red wine data set}
\noindent
Figure~\ref{fig:red_corr} presents histogram representations of marginal
posterior probability densities of some partial correlation parameters learnt
given the standardised red wine data; the trace of the joint posterior of all
the partial correlation parameters is also included. Figure~\ref{fig:red_corr}
on the other hand presents the marginals of some of the variance parameters.

\begin{figure}[!hb]
\centering
\includegraphics[width=13cm,height=8cm]{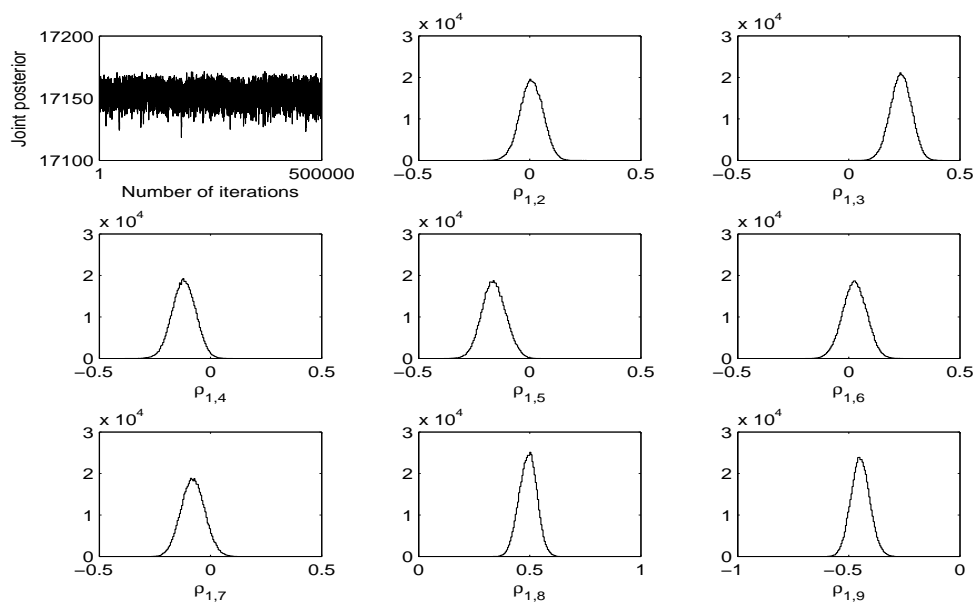}
\caption{The marginal posterior of some of the partial correlation parameters
$\rho_{ij}$ computed using the elements of the correlation matrix $\bSigma^{(red)}_S$ that
is updated in the first block of our MCMC chain, run with the red wine data ${\bf D}_S^{(red)}$ of Portuguese red wine samples; $i\neq j;\: i,j=1,\ldots,p=12$. The top left
hand panel of this figure presents the trace of the joint posterior
probability density of the elements of the upper triangle of
$\bSigma^{(red)}_S$.
} 
\label{fig:red_corr}
\end{figure}

\begin{figure}[!t]
\centering
\includegraphics[width=13cm,height=8cm]{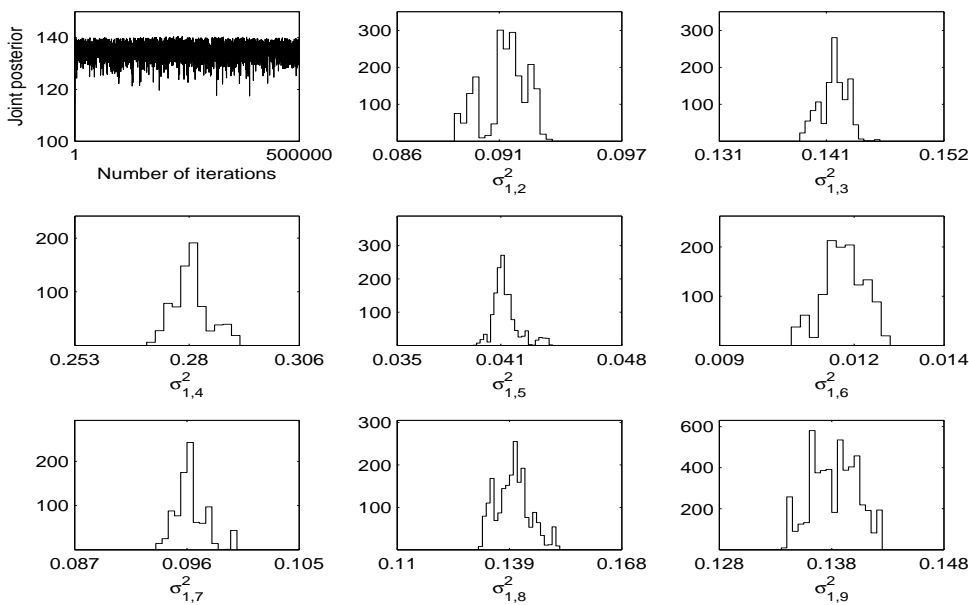}
\caption{
The
upper panel of this figure presents the trace of the joint
posterior probability of the $G_{ij}$ parameters and the variance parameters
$\sigma_{ij}^2$ (of the Normal likelihood) used in this second block
update of our MCMC chain, run with the red wine data ${\bf D}_S^{(red)}$ of Portuguese red wine samples; $i\neq j;\: i,j=1,\ldots,p=12$. The marginal of some of the variance parameters are also shown in the other panels of this figure.
} 
\label{fig:red_var}
\end{figure}


\section{Comparing our learnt results against empirical and regression analysis of red-wine data}
\label{sec:intro}
\noindent
The data on 1599 samples of Portuguese red wines is discussed by \ctn{CorCer09} and
considered in the main paper (Section~4.2). The between-columns correlation
structure and graphical model of this data are reported in this section. These results are reviewed in light of independent data analysis of the red wine data that we undertook. The original red wine data is ${\bf D}_{orig}^{(red)}$, of which ${\bf D}_S^{(red)}$ is a standardised subset. The dataset has 12 columns,
that contain information
on vino-chemical attributes of the sampled wines; these properties are
assigned the following names: ``fixed acidity'' ($X_1$), ``volatile acidity''
($X_2$), ``citric acid'' ($X_3$), ``residual sugar'' ($X_4$), ``chlorides''
($X_5$), ``free sulphur dioxide'' ($X_6$), ``total sulphur dioxide'' ($X_7$),
``density'' ($X_8$), ``pH'' ($X_9$), ``sulphates'' ($X_{10}$), ``alcohol''
($X_{11}$); the 12-th column is the assessed ``quality'' ($X_{12}$) of a wine in the sample. The standardised version of variable $X_i$ is $Z_i$, $i=1,\ldots,12$.

A matrix of scatterplots of $X_j$ against $X_i$ is shown in
Figure~\ref{fig:scatter}, for $i=1,\ldots,11$. These scatterplots visually
indicate moderate correlations between the following pairs of variables: fixed
acidity-citric acid, fixed acidity-density, fixed acidity-pH, volatile
acidity-citric acid, free sulphur dioxide-total sulphur dioxide,
density-alcohol. All these variables share an edge at probability $\geq 0.05$
in our learnt graphical model of data ${\bf D}_S^{(red)}$ (Figure~3 of main
paper). We note that all moderately correlated variable pairs, as represented
in these scatterplots, are joined by edges in our learnt graphical model of
the red wine data--as is to be expected if the learning of the graphical model
is correct. Such pairs include fixed acidity-citric acid, fixed
acidity-density, fixed acidity-pH, volatile acidity-citric acid, free sulphur
dioxide-total sulphur dioxide, density-alcohol. However, an edge may exist
between a pair of variables even when the apparent empirical correlation
between these variables is low (see Section~\ref{sec:white}); this owes to the
effect of other variables. 
However, an edge may exist between a pair of variables even
when the apparent empirical correlation between these variables is low (see
Section~\ref{sec:white}), owing to the effect of other variables. 
Noticing such edges from the residual-sugar variable, 
we undertake a regression analysis (ordinary least squares) with
residual-sugar regressed against the other remaining 10 vino-chemical
variables. The MATLAB output of that analysis is included in
Figure~\ref{fig:sugar}. The analysis indicates that the covariates with
maximal (near-equal) effect on residual-sugar, are density and alcohol;
residual-sugar is learnt to enjoy an edge with both density ($Z_7$) and
alcohol ($Z_{10}$) in our learnt graphical model of the red wine data (Figure~3 of WC).

We also undertook a separate ordinary least squares analysis with the response variable quality, regressed against the vino-chemical variables as the covariates. The MATLAB output of this regression analysis in in Figure~\ref{fig:quality}. We notice that the strongest (and nearly-equal) effect on quality is from the variables volatile-acidity and alcohol--the very two variables that share an edge at probability $\geq 0.05$ with quality, in our learnt graphical model of the red wine data.

\section{Cholesky Factorisation and Matrix Inversion by Forward Substitution}


Let a $p\times  p$-square positive-definite (correlation) matrix be $\bSigma_{C}^{(S)}=\bL_{C}^{(S)} (\bL_{C}^{(S)})^T$. 
The Cholesky factorisation of $\bSigma_{C}^{(S)}=[s_{ij}]$ into its unique square root
$\bL_{C}^{(S)}=[l_{ij}]$ can be shown to be defined by the following scheme:

\begin{eqnarray}
l_{11} &=& \sqrt{s_{11}},\nonumber \\
l_{i1} &=& \displaystyle{\frac{s_{i1}}{l_{11}}},\quad{i=1,\ldots,p},\nonumber \\
l_{ij} &=& \displaystyle{\frac{\sqrt{s_{ij} - \sum\limits_{k=1}^{j-1} l_{ik} l_{kj}}}{l_{jj}}}\quad{j=1,\ldots,i-1;\:i=1,\ldots,p},\nonumber \\
l_{ii} &=& \displaystyle{\sqrt{s_{ii} - \sum\limits_{k=1}^{i-1} l_{ik}^2}}\quad{i=2,\ldots,p},\nonumber \\
\end{eqnarray}
while forward substitution seeks $\bL_C^{-1}$ s.t. $\bL_C \bL_C^{-1}=\bI$, where $\bI$ is the $pXp$-dimensional identity matrix. Then the scheme for forward substitution is the following:
\begin{eqnarray}
m_{11} &=& \displaystyle{\frac{1}{l_{11}}},\nonumber \\
l_{i1} &=& \displaystyle{\frac{s_{i1}}{l_{11}}},\quad{i=1,\ldots,p},\nonumber \\
l_{ij} &=& \displaystyle{\frac{\sqrt{s_{ij} - \sum\limits_{k=1}^{j-1} l_{ik} l_{kj}}}{l_{jj}}}\quad{j=1,\ldots,i-1;\:i=1,\ldots,p},\nonumber \\
l_{ii} &=& \displaystyle{\sqrt{s_{ii} - \sum\limits_{k=1}^{i-1} l_{ik}^2}}\quad{i=2,\ldots,p},\nonumber \\
\end{eqnarray}

\begin{figure}[!ht]
\centering
\includegraphics[width=18cm,height=18cm]{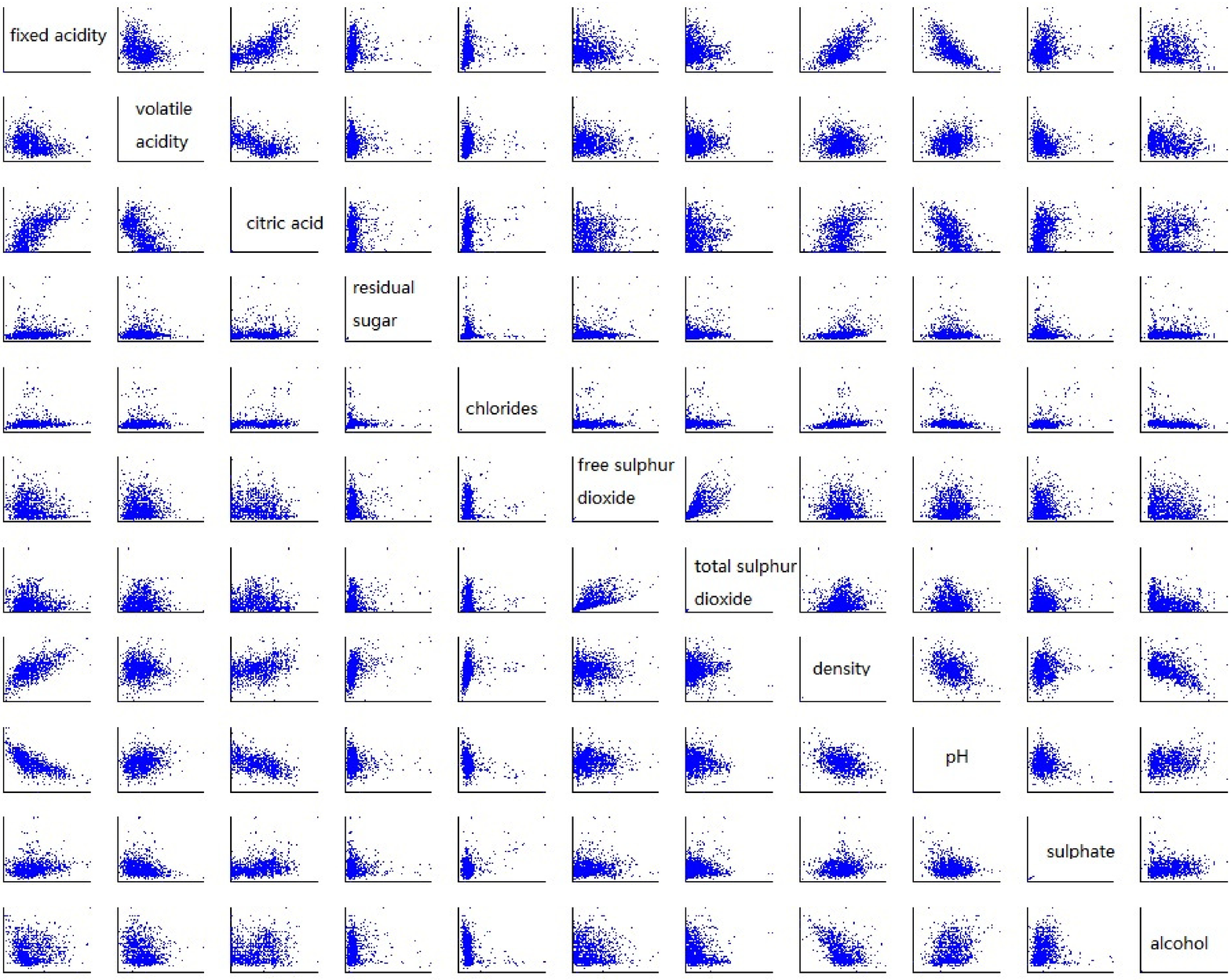}
\caption{Matrix of scatterplots of the 11 different vino-chemical variables $X_1,\ldots,X_{11}$ that form the first 11 columns of the red wine data ${\bf D}_{orig}^{(red)}$. Here $X_j$ is plotted against $X_i$, $i\neq j$, $i,j=1,\ldots,11$. The $X_i$ relevant to the $i$-th row is named in the diagonal element of the $i$-th row; $j$ increases from 1 to 11 from left to right.}
\label{fig:scatter}
\end{figure}

\begin{figure}[!ht]
\centering
\includegraphics[width=18cm,height=18cm]{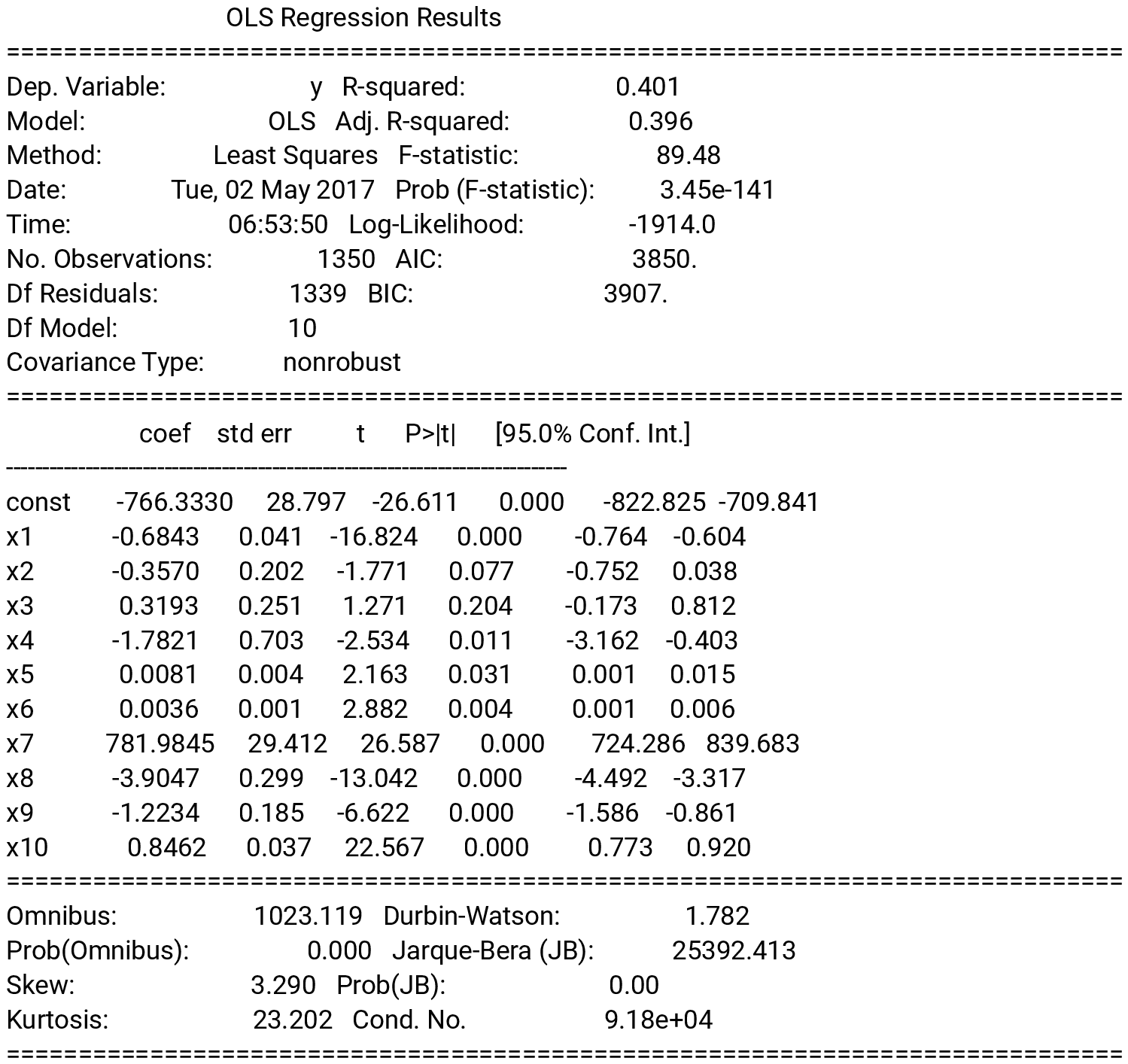}
\caption{Output of ordinary least square analysis of regressing residual sugar on the other 10 vino-chemical attributes in the red wine data.}
  \label{fig:sugar}
\end{figure}

\begin{figure}[!ht]
\centering
\includegraphics[width=18cm,height=18cm]{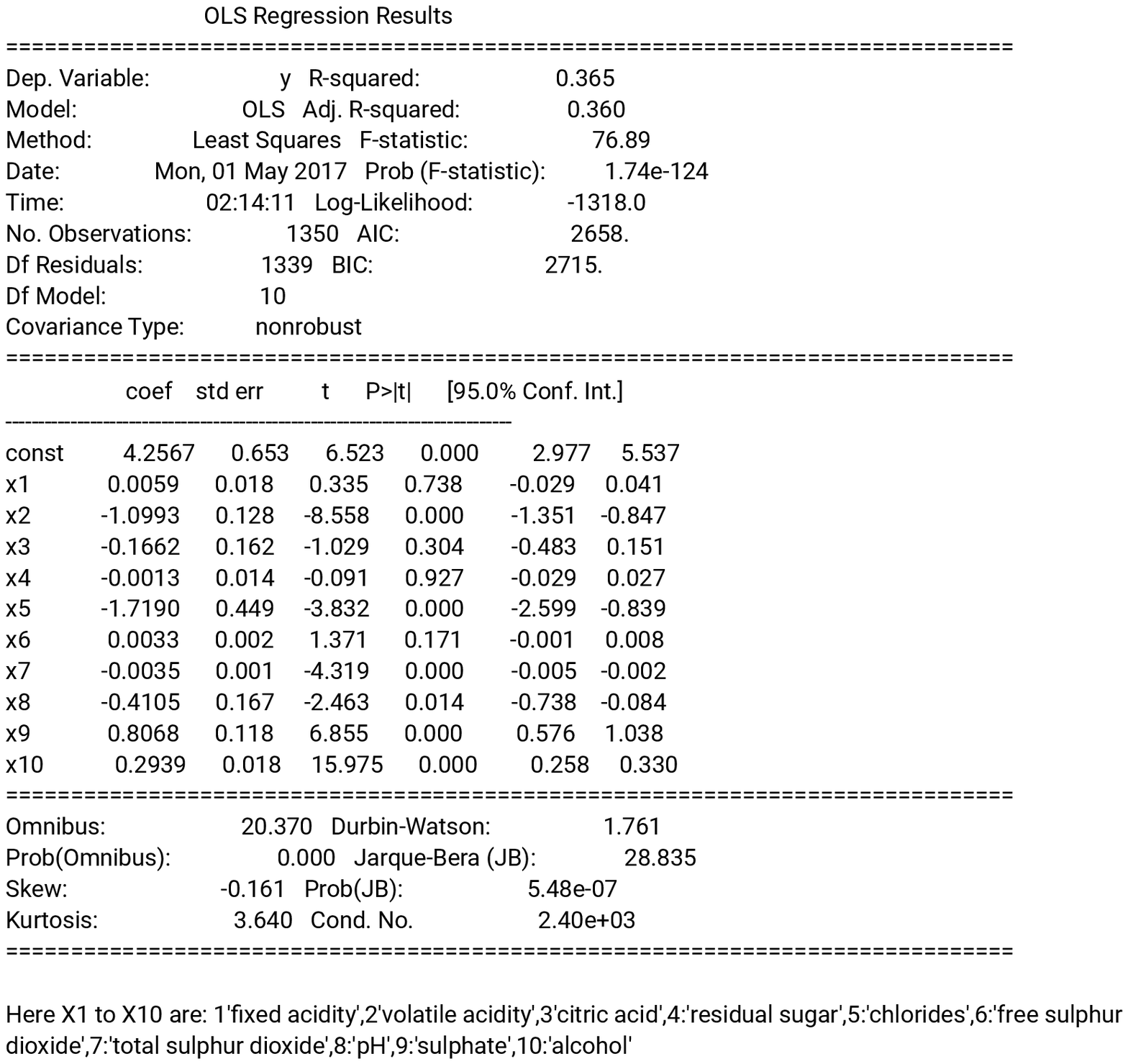}
\caption{Output of ordinary least square analysis of regressing quality on the vino-chemical attributes of red wine samples in the red wine data.}
  \label{fig:quality}
\end{figure}

\renewcommand\baselinestretch{1.}
\small
\bibliographystyle{ECA_jasa}

\end{document}